\documentclass[a4paper,UKenglish,cleveref, autoref, thm-restate,pdfa]{lipics-v2021}

\pdfoutput=1 
\hideLIPIcs  

\title{Deterministically Simulating Barely Random Algorithms in the Random-Order Arrival Model} 

\titlerunning{Deterministically Simulating Barely Random Algorithms in the Random-Order Arrival Model} 

\author{Allan Borodin}{University of Toronto, Toronto, ON, Canada \and \url{https://www.cs.toronto.edu/~bor/}}{bor@cs.toronto.edu}{}{}

\author{Christodoulos Karavasilis}{University of Toronto, Toronto, ON, Canada \and \url{https://www.cs.toronto.edu/~ckar/}}{ckar@cs.toronto.edu}{}{}

\author{David Zhang}{University of Toronto, Toronto, ON, Canada}{dzhang@cs.toronto.edu}{}{}

\authorrunning{A. Borodin, C. Karavasilis, and D. Zhang} 

\Copyright{Allan Borodin, Christodolous Karavasilis, and David Zhang} 

\ccsdesc[100]{Theory of computation~Online algorithms} 

\keywords{Online Algorithms, Random-Order Model, Random Bit Extraction, Barely Random Algorithms, Derandomization} 

\category{} 

\relatedversion{} 

\nolinenumbers 

\usepackage[ruled]{algorithm}
\usepackage{algpseudocode}
\usepackage{bm}
\usepackage{booktabs}
\usepackage[toc,page,header]{appendix}
\usepackage{minitoc}



\EventEditors{John Q. Open and Joan R. Access}
\EventNoEds{2}
\EventLongTitle{42nd Conference on Very Important Topics (CVIT 2016)}
\EventShortTitle{CVIT 2016}
\EventAcronym{CVIT}
\EventYear{2016}
\EventDate{December 24--27, 2016}
\EventLocation{Little Whinging, United Kingdom}
\EventLogo{}
\SeriesVolume{42}
\ArticleNo{23}

\renewcommand \thepart{}
\renewcommand \partname{}

\begin{document}

\doparttoc
\faketableofcontents

\maketitle

\begin{abstract}
Interest in the random-order model (ROM) leads us to initiate a study of utilizing random-order arrivals to extract random bits with the goal of derandomizing algorithms. Besides producing simple algorithms, simulating random bits through random arrivals enhances our understanding of the comparative strength of randomized online algorithms (with adversarial input sequences) and deterministic algorithms in the ROM. We consider three $1$-bit randomness extraction processes. Our best extraction process returns a bit with a worst-case bias of $2 - \sqrt{2} \approx 0.585$ and operates under the mild assumption that there exist at least two distinct items in the input.
We motivate the applicability of this process by using it to simulate a number of barely random algorithms for weighted interval selection (single-length with arbitrary weights, as well as monotone, C-benevolent and D-benevolent weighted instances), the proportional and general knapsack problems, job throughput scheduling, and makespan minimization. 

It is well known that there are many applications where a deterministic ROM algorithm significantly outperforms any randomized online algorithm (in terms of competitive ratios). The classic example is that of the secretary problem.  We ask the following fundamental question: Is there any application for which a randomized algorithm outperforms any deterministic ROM algorithm? Motivated by this question, we view our randomness extraction applications as a constructive approach toward understanding the relationship between randomized online algorithms and deterministic algorithms in the ROM. 
\end{abstract}

\newpage

\section{Introduction}
\label{sec:intro}
In the basic deterministic online model, input sequences are determined adversarially. The objective in competitive analysis is then to minimize the maximum (over all possible input sequences) {\it competitive ratio} between an optimal solution  and the algorithm's solution (in terms of some objective function). For randomized online algorithms, the usual model assumes an oblivious adversary; that is, an adversary chooses an input based on knowledge of an online algorithm and the random distributions it uses (but not any of the random draws from those distributions). The competitive analysis objective is now to minimize (over all input sequences) the ratio between the value of an optimum solution and the expectation of the algorithm's solution (where the expectation is in terms of the algorithm's random choices). In the deterministic random-order arrival model (ROM), an adversary chooses a multiset of input items, and then these items arrive in random order. With respect to the ROM model, the goal in competitive analysis is to minimize (over all adversarial input sets) the ratio between an optimum solution and the expectation of the algorithm's solution (where the expectation here is in terms of the random permutation of the input items).  

Returning to the fundamental question raised in the abstract, as far as we know it could be that for any online problem in the framework of request answer games \cite{Ben-DavidBKTW94}, if there is an algorithm $\proc{Alg}$ with competitive ratio $\rho$ then there is an algorithm $\proc{Alg}'$ whose ROM competitive ratio is $\rho$ (or more generally $f(\rho)$ for some function $f$). Such a result could be proved non-constructively (analogous to the relationship between randomized algorithms against an adaptive adversary and deterministic algorithms). Constructively, can we always simulate a randomized algorithm (with little or no loss in the competitive ratio) by a deterministic algorithm in the ROM model? In this regard, we have a very limited but still interesting conjecture: Any ``1-bit barely random online algorithm''  can be simulated with little loss by a deterministic  algorithm in the ROM model. We will show how to take advantage of the randomness in the arrival order to extract a single random bit (sometimes with a small bounded bias) to simulate a variety of 1-bit barely random online algorithms\footnote{In all our applications, we construct derandomized algorithms in the ROM whose competitive ratios satisfy $f(\rho)=c \cdot \rho$ for some small fixed constant $c\geq 1$.}. With the exception of the online identical-machine makespan problem, the simulating ROM algorithm provides the first known ROM competitive ratio or improves upon the currently best known competitive ratio in the ROM model. Our applications include several ``real-time scheduling'' algorithms in which case we need to define what we mean by a real-time algorithm in the ROM model.  

It is well understood that ROM ratios often provide a much more realistic measure of real world performance. It is also well known (see Kenyon \cite{Kenyon96}, Karande et al. \cite{Karande2011} and Gupta and Singla \cite{GuptaS20}) that any algorithm having competitive ratio $\rho$ in the ROM model is also a $\rho$-competitive algorithm when the online input items are drawn i.i.d. (independently and identically) from a distribution even if the distribution is not known.

\subsection{Our Results}
We initiate a study of utilizing random-order arrivals to extract random bits with the goal of derandomizing algorithms. 
While  randomness is implicitly being extracted (by sampling some initial inputs) in all secretary applications\footnote{The ``secretary'' terminology refers to random-order algorithms without revoking. The implicit extraction of randomness has also been utilized in \cite{PenaB19} where it is shown that disjoint copies of a $2 \times 2$  bipartite graph cannot be used to derive an asymptotic lower bound for matching.}, our randomness extraction is explicit and yields a random bit. 
In particular, we consider 1-bit  barely random algorithms that use one unbiased bit $b$ (i.e., equal probability that $b = 0$ or $b=1$)  to randomly choose between two deterministic online algorithms $\proc{Alg}_1$ and $\proc{Alg}_2$  with the property that for every input sequence $\sigma$, the output of $\proc{Alg}_1(\sigma) + \proc{Alg}_2(\sigma) \geq \proc{Opt}(\sigma) / c$ so that at least one of these two algorithms will provide a ``good''  approximation to an optimal solution on $\sigma$. We define the bias $\beta \geq \frac{1}{2}$ of a random bit $b$ to mean\footnote{We can also state the bias in terms of the $\Pr[b=0]$ in the same way. The extraction application can define whether we want the bias of a bit with $b=1$ or $b=0$.}  $\Pr[b=1] \in [\frac{1}{2},\beta]$  so that for an unbiased bit, $\beta = \frac{1}{2}$. Furthermore, these deterministic algorithms usually perform optimally or near optimally when $\sigma$ is a sequence of {\it identical} items. We first present a simple process that extracts a bit from a random-order sequence with a worst case bias of $\frac{2}{3}$ when there are at least two distinct items. We then study another simple process that extracts an unbiased bit if all input items are distinct, and then a combined process that will extract a random bit with a worst case bias of $2-\sqrt{2}\approx 0.585$ under the assumption that not all items are identical. We use this last process to derandomize a number of $1$-bit barely random algorithms. In particular, we derandomize the general and proportional knapsack algorithms of Han et al. \cite{han2015randomized}, the interval selection algorithms of Fung et al. \cite{fung2014improved} for single-length and monotone instances with arbitrary weights as well as C-benevolent and D-benevolent weighted instances, the unweighted throughput scheduling algorithms of Kalyanasundaram and Pruhs \cite{kalyanasundaram2003maximizing} and Chrobak et al. \cite{chrobak2007online}, and the Fung et al. \cite{fung2014improved} algorithm for single-length throughput instances with arbitrary weights. We refer the reader to Section \ref{sec:intervals} for the definitions of monotone, C-benevolent and D-benevolent weighted instances of interval selection. 
We also consider derandomizing the algorithm of Albers \cite{Albers02} for the makespan problem on identical machines, but here our algorithm is not as good as the random-order algorithm of Albers and Janke \cite{AlbersJ21}. Finally, we consider the algorithm of Kimbrel and Saia \cite{KimbrelS00} for the two-machine job shop problem, where we obtain no improvement over what is already known for deterministic algorithms under adversarial order.

With the exception of single-length interval selection, results for variants of the throughput problem are naturally in the ``real-time'' input model. In this model, problems are stated in terms of a global clock time $t$ with the assumption that for all input items $i$ and $j$, $i < j$ implies $r_i \leq r_j$, where $r_i$ is the release time of the $i^{th}$ job (i.e., the time at which the job arrives in the notation of scheduling theory). We assume a clairvoyant model where all parameters of a job are known at release time. If we apply a random permutation in a real-time model, we will lose the real-time assumption and be back in the standard online framework. To maintain the real-time assumption, we define the {\it real-time with random-order model} to mean that the set of release times in the set of input items are fixed and a random permutation is then applied to the remainder of each input item. As motivation, consider servicing periodic requests whose requirements may vary. In data centers, batches of jobs may arrive at known times, but the resources required to process each batch can depend on its contents or other external factors. Similarly, the buying and selling of stocks occurs at discrete time units often with specified limits on the time to execute a transaction. We will precisely specify this real-time model for the interval selection and throughput problems.




In scheduling theory, preemption of a job usually means preemption with {\it resuming} where interrupted jobs can be continued from the time of the interruption. Another real-time model allows for preemption with {\it restarting} where interrupted jobs must be completely restarted. The most restrictive version of preemption is the {\it revoking} model where interrupted jobs must then be discarded. We note that restarting (which includes the more restrictive revoking of jobs) provides no benefit to an optimal solution, whereas in the resuming model an optimal solution may also exploit the power of preemption. Hence, competitive ratios in the resumption model can be worse than in the restarting model. While resuming of jobs is only considered in the real-time model, both restarting and revoking can apply to the standard online model. To the best of our knowledge, restarting has only been used in the real-time model whereas revoking has mainly been applied in the standard online model. 

For most of the applications we consider, the randomized algorithms being simulated allow for revoking of previously accepted items and our ROM algorithms will also allow such revoking. Even when a randomized algorithm being simulated does not use revoking, we may still need to allow such revoking since deterministic algorithms can usually be forced into some initial bad decisions, whereas randomized algorithms avoid these bad decisions on average. To the best of our knowledge, with the exception of the knapsack and identical-machine makespan problems, none of the aforementioned applications have been studied in the random-order model. We also indicate that we can derandomize the 1-bit barely random makespan algorithm of Albers \cite{Albers02} for the makespan problem on identical machines but the competitive ratio of our derandomized algorithm is worse than the ROM algorithm of Albers and Janke \cite{AlbersJ21}. Regardless of whether our ratios improve upon those of existing ROM algorithms, we view all our derandomizations of 1-bit barely random algorithms as a ``proof of concept'' that constructively shows that there exists deterministic constant competitive algorithms in the random-order model. We note that unlike secretary style algorithms, the barely random algorithms we simulate do not need knowledge of $n$, the number of input items. Establishing optimal ROM ratios is then the next challenge.

Finally, if we make the strong assumption that all input items are distinct (i.e., no input item appears more than once), we indicate how we can extract more than one random bit.  Our 1-bit extraction only uses the very mild but necessary assumption that not all items are identical\footnote{There is no entropy in a random-order input when all input items are identical.}, noting that it is usually easy to analyze the performance of an online algorithm when all items are identical.

In the following table, we consider problems for which there are $1$-bit barely random algorithms that improve upon the deterministic competitive ratio. \textbf{Det.-LB} and \textbf{Rand.-LB} denote the best known deterministic and randomized lower bounds under adversarial order, respectively. \textbf{Rand.-UB} denotes the best known randomized upper bound under adversarial order, and \textbf{Det.-ROM} denotes the competitive ratio obtained by applying our bit extraction technique to derandomize algorithms in the ROM. With the exception of the makespan problems, our derandomizations strictly beat deterministic lower bounds under adversarial order. Bold entries correspond to results of this paper. For makespan problems, we use $n$ to denote the number of jobs in the input and $m$ to denote the number of available machines. We defer the definitions of terms used in Table \ref{table:table-of-results} to Section \ref{sec:applications}, where we specify each problem in detail.

    \begin{table}[h!]
    \centering
    \resizebox{\textwidth}{!}{%
    \begin{tabular}{lcccc}
    \toprule
    \textbf{Problem} & \textbf{Det.-LB} & \textbf{Rand.-LB} & \textbf{Rand.-UB} & \textbf{Det.-ROM} \\
    \midrule
    Knapsack & & &  &  \\
    \>\>\>\> General weights with revoking. & $\dagger$ \cite{iwama2010online} & $1+\frac{1}{e} \approx 1.368$ \cite{han2015randomized} & $2$ \cite{han2015randomized} & $\bm{\frac{1}{\sqrt2-1} \approx 2.414}$ \\
    \>\>\>\> Proportional weights without revoking. & $\dagger$ \cite{Marchetti-SpaccamelaV95} & $2$ \cite{han2015randomized} & $2$ \cite{han2015randomized} & ${\bm \dagger}$\\
    \>\>\>\> Proportional weights with revoking. & $\frac{1+\sqrt5}{2} \approx 1.618$ \cite{iwama2002removable} & $\approx 1.270$ \cite{hachler2025untrusted} & $\frac{10}{7}\approx 1.429$ \cite{han2015randomized} & $\bm {\frac{5}{3\sqrt2-1}\approx 1.542}$ \\
    Interval selection with revoking & & & & \\
    \>\>\>\> Equal-length intervals with \\\>\>\>\> arbitrary weights. & $4$ \cite{woeginger1994line} & $1+\ln 2 \approx 1.693$ \cite{epstein2008improved} & $2$ \cite{fung2014improved} & $ \bm{ \frac{1}{\sqrt2-1}}$ \\
    \>\>\>\> Monotone intervals with\\\>\>\>\> arbitrary weights. & $4$ \cite{woeginger1994line} & $1+\ln 2 \approx 1.693$ \cite{epstein2008improved} & $2$ \cite{fung2014improved} & $ \bm{\frac{1}{\sqrt2-1}} $ \\
    \>\>\>\> Any-length intervals with\\\>\>\>\> C-benevolent weights. & $4$ \cite{woeginger1994line} & $1+\ln 2 \approx 1.693$ \cite{epstein2008improved} & $2$ \cite{fung2014improved} & $ \bm{ \frac{1}{\sqrt2-1}}$ \\
    \>\>\>\> Any-length intervals with\\\>\>\>\> D-benevolent weights. & $3^{\S}$ \cite{woeginger1994line} & $\frac{3}{2}^{\S}$ \cite{epstein2008improved} & $2$ \cite{fung2014improved} & $\bm {\frac{1}{\sqrt2-1}}$ \\
    Single-machine job throughput & & & & \\
    \>\>\>\> Any-length jobs with\\\>\>\>\> unit weights and resuming. & $\dagger$ \cite{baruah1994line} & $\Omega(1)$ & $\approx 258048$ \cite{kalyanasundaram2000optimal} & $\bm{\approx 311491}$\\
    \>\>\>\> Equal-length jobs with\\\>\>\>\> unit weights (and no recourse.) & $2$ \cite{goldman2000online} & $\frac{4}{3}$ \cite{goldman2000online} & $\frac{5}{3}$ \cite{chrobak2007online} & $\bm{\frac{5}{2\sqrt2}\approx 1.768}$\\
    \>\>\>\> Equal-length jobs with\\\>\>\>\>
     arbitrary weights and restarting. & $4$ \cite{woeginger1994line} & $\frac{6}{5}$ \cite{chrobak2007online} & $3$ \cite{fung2014improved} & $\bm{\frac{3}{2\sqrt2-2} \approx 3.621}$ \\
    2-machine job shop with resuming. & $2-\frac{1}{n}$ \cite{KimbrelS00} & $\frac{3}{2}-\frac{1}{2n}$ \cite{KimbrelS00} & $\frac{3}{2}$ \cite{KimbrelS00} & $\bm{10-5\sqrt2\approx 2.930}$\\
    Identical-machines makespan. & $1.88^{\P}$ \cite{Rudin2001ImprovedBounds} & $\frac{1}{1-(1-\frac{1}{m})^m} \xrightarrow[]{m \to \infty} \frac{e}{e-1}$ \cite{chen_vanVliet_woeginger_1994_lower_bound_randomized_online_scheduling,sgall1997lower} & $1.916$ \cite{Albers02} & $\bm {1.916(4-2\sqrt2)+1\approx 3.245}$ \\
    \bottomrule
    \end{tabular}}
    \caption{\footnotesize\\$^{\dagger}$ No constant ratio is possible.\\
    $^{\S}$ This result holds for D-benevolent instances whose weight functions are surjective on $\mathbb{R}_{\geq 0}$. Most lower bounds for D-benevolent instances only hold for functions that satisfy the same condition.
    \\$^{\P}$ This bound holds for
    for $m \ge 3600$ machines; many (better) bounds exist for fewer machines. 
    }
    \label{table:table-of-results}
    \end{table}

To the best of our knowledge, we are the first to study the knapsack problem with revoking (sometimes called removal) in the ROM. Our result is an improvement over the current best $6.65$-competitive randomized ROM algorithm for the general knapsack problem due to Albers, Khan and Ladewig \cite{AlbersKhanLadewig2021}. We note that their algorithm requires advance knowledge of the number of input items whereas we do not. However, their algorithm is {\it without revoking} whereas our algorithm requires revoking.

Borodin and Karavasilis \cite{BorodinK23} give a deterministic $2.5$-competitive algorithm for unweighted interval selection with arbitrary lengths and revoking in the ROM. Since D-benevolent weights include the unweighted case, by derandomizing the D-benevolent interval selection algorithm of Fung et al. \cite{fung2014improved}, we improve this bound, achieving a deterministic $\frac{1}{\sqrt2-1}\approx 2.414$-competitive algorithm. We note that the Fung et al. algorithm is only for the real-time model whereas the Borodin and Karavasilis algorithm is for the more general online model. Finally, we note that the lower bound of 4 due to Woeginger \cite{woeginger1994line} for equal-length intervals with arbitrary weights follows from a real-time construction, while the algorithm of Fung et al. \cite{fung2014improved} is 2-competitive even in the general online model.

With the exception of the result of Borodin and Karavasilis for unweighted interval selection, we establish the first competitive ratios in the ROM for the aforementioned variants of interval selection and unweighted single-machine throughput maximization. We also initiate the study of the two-machine job shop problem in the ROM by derandomizing the algorithm of Kimbrel and Saia \cite{KimbrelS00}, though we do not improve upon the tight ratio of 2 for deterministic algorithms under adversarial order. Albers and Janke \cite{AlbersJ21} present a deterministic $1.848$-competitive algorithm for the identical-machines makespan problem in the ROM, which beats the ratio obtained by our derandomization.

\section{Preliminaries}
\label{sec:preliminaries}
In the basic online model, input items arrive sequentially. The set of items and the order in which they arrive is determined by an adversary. When the $i^{th}$ input item arrives, an immediate irrevocable deterministic (respectively, randomized) decision must be made based on the first $i$ items without knowledge of any remaining items. We will refer to such algorithms as deterministic (respectively, randomized) online algorithms. When defining randomized algorithms, we have to clarify the nature of the adversary. As is most standard, we will assume an oblivious adversary. An oblivious adversary knows the algorithm and the distributions it is using to make decisions, but does not know the outcomes for any sampling of these distributions. We will mainly be considering barely random algorithms, which only use one unbiased random bit. We will also discuss barely random algorithms that use a small constant number of unbiased random bits. The length $n$ of the input stream may be known or not known \textit{a priori}. We refer the reader to \cite{BorOnlineBook} for a general introduction to online algorithms and to \cite{Ben-DavidBKTW94} for a discussion of the relative power of the different adversaries.

The basic model has been extended in two important ways. First, online algorithms {\it with recourse} extend the basic model by allowing some limited way for decisions to later be modified. The simplest form of recourse allows accepted or scheduled input items to be permanently deleted. These algorithms  will be referred to as online algorithms {\it with revoking}. Second, online algorithms are extended to relax the assumption that the input stream is completely adversarial. In addition to the online input model, we will consider the real-time model and the random-order model (ROM).  The real-time model is the basic model used in scheduling algorithms. In the real-time model, every input item includes a time $r_i$ when the $i^{th}$ item is released or arrives and make the assumption that $r_i \leq r_{i+1}$ for all $i \geq 1$. See Lawler et al. \cite{LawlerLKS93} for a general reference to scheduling algorithms. 

In the random-order arrival model, the adversary chooses the multiset of input items which then arrive in random order.
There is a large body of work concerning deterministic and randomized algorithms in the random-order model (mainly without revoking) dating back to at least the secretary problem. See Gilbert and Mosteller \cite{GilbertM66} for a history of the secretary problem.
If we apply a random permutation to a set of real-time instances, we will no longer be in the real-time model since the release times are being permuted. However, we will show how to effectively merge the real-time and ROM models in our applications for interval and throughput scheduling.

With the exception of the makespan (including job shop)
problems, our applications are maximization problems. Specifically, we consider packing problems where any subset of a feasible solution is feasible. Hence, we can always revoke any previously accepted or scheduled input items and maintain feasibility.

We consider the {\it competitive ratio} of an algorithm as a measure of performance.
For maximization problems, there is no fixed convention as to whether or not these ratios should be stated as being greater than 1 or less than 1. To maintain consistency throughout the paper, we adopt the convention that ratios be greater than 1 but note that results for the knapsack problem are sometimes given as ratios less than 1. When discussing each application we will define each problem and provide related references. We note that all of our applications have been widely studied in various settings and we are only providing the most relevant literature.

In defining competitive ratios, we follow the convention that competitive ratios are asymptotic.
For a maximization problem, if $\proc{Opt}(x)$ is the value of an optimal solution and $\proc{Alg}(x)$ is the value produced by an algorithm on input $x$,
the competitive ratio of a deterministic online algorithm is defined as $\proc{Opt}(x) \leq \rho \cdot \proc{Alg}(x) + o(\proc{Opt}(x))$.
We say that $\rho(\proc{Alg})$ is a strict competitive ratio when there is no additive term. 


For randomized algorithms, we modify the definition by taking the expectation ${\mathbb E}[\proc{Alg}(x)]$ over the random choices of the algorithm. For a ROM algorithm, we again take the expectation ${\mathbb E}[\proc{Alg}(x)]$ where now the expectation is over the random permutations of the input sequence. Our applications will take a barely random  online algorithm $\proc{Alg}$ with (perhaps strict) competitive ratio $\rho$ and create a deterministic ROM algorithm $\proc{Alg}'$ with asymptotic competitive ratio $\rho'$ that is within a small constant factor of $\rho$.  The precise ratios for $\rho'$ can  depend on the application.  Once again, we emphasize that we see these results as a proof of concept in that most of our deterministic ROM ratios are, to the best of our knowledge, either the first known ROM results or improve upon previously known results in some way. Furthermore, if the given randomized algorithm is ``conceptually simple'', then the ROM algorithm will also be conceptually simple although the analysis in some of our applications may not be. We see these results as a small step towards understanding the relationship between randomized online and deterministic ROM results.

\section{Randomness Extraction}
\label{sec:randomness-extraction}

 We recall our motivating question with regard to online algorithms; namely, what  is the power of random-order deterministic algorithms relative to adversarial order randomized algorithms? We know that there are a number of problems where deterministic random-order algorithms provide provably better (or at least as good) competitive ratios than randomized algorithms with adversarial order (e.g., the secretary problem, bin packing \cite{Chandra1992,HebbarKS2024,Kenyon96}, and bipartite matching \cite{Mehta2013}). But are there problems where adversarial order randomized algorithms are provably (or at least seemingly) better than random-order deterministic algorithms? It is natural then to see if we can use the randomness in the arrival of input items to extract random bits. Such bits may then be used to derandomize certain algorithms when assuming random-order arrivals. \textit{Barely random} algorithms \cite{reingold1994randomized} use a (small) constant number of random bits, and are well suited to be considered for this purpose. These algorithms are often used in the \textit{classify and randomly select} paradigm, where inputs are partitioned into a small number of classes, and the algorithm randomly selects a class of items to work with.

Our focus will be on simple $1$-bit extraction algorithms, which we can apply to $1$-bit barely random algorithms. We briefly consider a way of extracting two biased bits for derandomizing an algorithm for two-length interval selection with arbitrary weights, and discuss possible extensions to extracting multiple bits from random-order sequences in Section \ref{sec:beyond-1-bit}. 



 Our first randomness extraction process\footnote{Our extractors sit in the background of algorithms that request the random bit. For this reason, we find it more natural to refer to them as processes.} is described in Algorithm \ref{alg:proc-1}. We assume that there exist at least two different classes, or item types, that all input items belong to. The item types used by Process \ref{alg:proc-1} are defined by the first item that arrives. Items that are identical to the first item are of the same type, while any distinct item is of a different type. Furthermore, we maintain a counter that represents the number of items that have arrived so far. Our process returns $1$ if the first item of the second type arrives when the counter is even, and returns $0$ otherwise. The adversary can choose the frequency of each type to affect the probability of getting either output, but we show that the worst case bias of this bit is $\frac{2}{3}$.

\begin{algorithm}
\caption{Process 1: Biased bit extraction}\label{alg:proc-1}
\begin{algorithmic}
\State On the arrival of $I_i$:
\If{$i = 1$}
    \State $type \gets type(I_1)$
    \Else
        \If{$type \neq type(I_i)$}
            \State \Return $1-(i \mod 2)$ and terminate the process.
        \EndIf
\EndIf
\end{algorithmic}
\end{algorithm}
\begin{theorem}
    The bit $b$ extracted by Process \ref{alg:proc-1} has a worst case bias of $\frac{2}{3}$. That is, $\Pr[b=1] \in (\frac{1}{2},\frac{2}{3}]$.
\end{theorem}
\begin{proof}
    Let there be $A$ items of $type_{A}$ and $B$ items of $type_{B}$, with $N = A + B$. Let $E_{v}$ be the event where the second item type arrives on an even counter. Let also $A_{e}$ (respectively, $B_{e}$) be the event that the first appearance of $type_{A}$ (respectively, $type_{B}$) is on an even counter, and $F_{A}$ (respectively, $F_{B}$) be the event that the very first item that arrives is of $type_{A}$ (respectively, $type_{B}$).
    We assume that $N$ is very large and that we are sampling from an infinite population. We have that:
    \[\Pr[E_{v}]= \Pr[B_{e}|F_{A}]\cdot \Pr[F_{A}] +  \Pr[A_{e}|F_{B}]\cdot \Pr[F_{B}]\]
    We start by computing $\Pr[B_{e}|F_{A}]$:
    \[ \Pr[B_{e}|F_{A}]= \frac{B}{N} + \left(\frac{A}{N}\right)^{2} \cdot \frac{B}{N} + \left(\frac{A}{N}\right)^{4} \cdot \frac{B}{N} + \left(\frac{A}{N}\right)^{6} \cdot \frac{B}{N} +\dots = \frac{B}{N}\sum_{i=0}^{+\infty}\left(\frac{A}{N}\right)^{2i}\]
and therefore:
\[\Pr[B_{e}|F_{A}]\cdot \Pr[F_{A}] = \frac{AB}{N^{2}}\sum_{i=0}^{+\infty}\left(\frac{A}{N}\right)^{2i}\]
Similarly, we get that: 
\[\Pr[A_{e}|F_{B}]\cdot \Pr[F_{B}] = \frac{AB}{N^{2}}\sum_{i=0}^{+\infty}\left(\frac{B}{N}\right)^{2i}\]
Putting everything together:
\[\Pr[E_{v}]= \frac{AB}{N^{2}}\sum_{i=0}^{+\infty}\left(\frac{A^{2i} + B^{2i}}{N^{2i}}\right)\]
Let $A=\alpha N$ with $\alpha \in (0,1)$. We can rewrite $\Pr[E_{v}]$ as follows:
\begin{align*}
    \Pr[E_{v}]&=\frac{\alpha(1-\alpha) N^{2}}{N^{2}} \sum_{i=0}^{+\infty}\frac{\left(\alpha N\right)^{2i}+(1-\alpha)^{2i}N^{2i}}{N^{2i}}\\\\
    &= \alpha (1-\alpha) \sum_{i=0}^{+\infty} \alpha^{2i} + (1-\alpha)^{2i}\\\\
    &= \alpha (1-\alpha) \frac{-2\alpha^2 + 2\alpha +1}{(\alpha - 2)\alpha(\alpha^2-1)}\\\\
    &= \frac{2\alpha^2 -2\alpha -1}{(\alpha +1)(\alpha -2)}
\end{align*}
Let $f(\alpha) = \frac{2\alpha^2 -2\alpha -1}{(\alpha +1)(\alpha -2)}$. We have that $f[(0,1)] = (\frac{1}{2},\frac{2}{3}]$. In conclusion, the worst case bias of the bit extracted through Process \ref{alg:proc-1} is $\frac{2}{3}$, which occurs when there are an equal number of two types.
\end{proof}


In all of our processes and applications we naturally assume that each input item is represented by a vector ${\bf x} = (x_1, \ldots , x_d)  \in {\mathbb R}^d$ for some $d \geq 1$. This allows us to strictly order distinct (i.e., not identical) items. That is, ${\bf x} < {\bf y}$ if and only if $x_i = y_i$ for $i < k$ and $x_k < y_k$ for some $k$ with $1 \leq k \leq d$. 
Under the assumption that all input items are distinct (and hence there exists a global ordering),  we are able to extract an unbiased bit simply by comparing the two first items to arrive (Algorithm \ref{alg:proc-2}). 
Notice that we can repeat this Process \ref{alg:proc-2} for the next pair of online items. Generally, given $2N$ online items we can extract $N$ unbiased bits.


\begin{algorithm}
\caption{Process 2: Distinct-Unbiased}\label{alg:proc-2}
\begin{algorithmic}
\State On the arrival of $I_1,I_2$:

\If{$I_1 < I_2$}
    \State \Return 1
\Else
    \State \Return 0
\EndIf
\end{algorithmic}
\end{algorithm}
\begin{theorem}
    Under the assumption that all input items are distinct (and hence there exists a global ordering among all items in the input instance), Process \ref{alg:proc-2} (Distinct-Unbiased) produces an unbiased bit.
\end{theorem}
\begin{proof}
    Let $I_1, I_2, ..., I_N$ be all the items in the instance such that $I_1 < I_2 <...<I_N$. Let $I_a, I_b$ be the first two items that arrive. Let $E$ denote the event that $I_a < I_b$. Let $F_i$ denote the event that item $I_i$ arrives first, and $S_i$ denote the event that $I_{i} < I_b$. We have that:
    \begin{align*}
        \Pr[E]&=\Pr[S_1 | F_1]\cdot \Pr[F_1]+ \dots + \Pr[S_N | F_N]\cdot \Pr[F_N]\\\\
        &=\frac{1}{N}\frac{N-1}{N-1} + \dots + \frac{1}{N}\frac{N-N}{N-1}\\\\
        & = \frac{1}{N}\sum_{i=1}^{N}\frac{N-i}{N-1} = \frac{1}{2}
    \end{align*}
\end{proof}


Motivated by the previous two extraction processes, we arrive at an improved $1$-bit extraction process. Under the mild assumption that there exist at least two distinct items, and we have an ordering among different items, we propose the following process: If the first two input items are identical, we use Process \ref{alg:proc-1}; otherwise we use the ordering between the first two distinct items types to determine the random bit. In particular, the random bit is 1 if the second type is less than the first type in the lexicographical ordering we defined on distinct items.  We call this process \texttt{COMBINE}. We will show that the worst-case bias of the returned bit is $2-\sqrt{2} \approx 0.585$. Let $\mathcal{I}$ be the set of input items, and $\sigma: \mathcal{I} \xrightarrow[]{} \mathbb{N}$ be the function dictating the ordering of different items. Furthermore, let $\mathbb{D}$ be the event where the first two input items are distinct, and $\mathbb{I}$ be the event where the first two items are identical. We begin by proving the following lemma:
\begin{lemma} \label{lem:half-distinct}
    $\Pr[\sigma(I_2) < \sigma(I_1) \;|\; \mathbb{D}] = \frac{1}{2}$, where $I_1,I_2$ are the first two items to arrive online.
\end{lemma}
\begin{proof}
    Let $d \leq N$ be the number of distinct items (or types) and let $C_1, C_2,...,C_d$ be the number of copies of each type, with $\sigma(type_i) < \sigma(type_j)$ when $i<j$. As before, we assume that $N$ is very large and that we are sampling from an infinite population. We have that:
    \[\Pr[\sigma(I_2) < \sigma(I_1)] = \sum_{i=1}^{d} \sum_{j = 1}^{i-1}\frac{C_i}{N}\frac{C_j}{N}\]
    Furthermore:
    \[\Pr[\mathbb{D}] = \sum_{i = 1}^{d} \sum_{\substack{j = 1\\ j\neq i}}^{d}\frac{C_i}{N}\frac{C_j}{N} = 2 \sum_{i = 1}^{d} \sum_{\substack{j = 1}}^{i-1}\frac{C_i}{N}\frac{C_j}{N} = 2 \cdot \Pr[\sigma(I_2) < \sigma(I_1)]\]
    This is because every pair $(i,j), i\neq j$ appears twice in the summation of $\Pr[\mathbb{D}]$. In conclusion:
    \[\Pr[\sigma(I_2) < \sigma(I_1) \;|\; \mathbb{D}] = \frac{\Pr[\sigma(I_2) < \sigma(I_1)]}{\Pr[\mathbb{D}]} = \frac{1}{2}\]
\end{proof}
We will now prove the following theorem.
\begin{theorem}
    The bit $b$ returned by \texttt{COMBINE} has a worst-case bias of $2-\sqrt{2}$. 
\end{theorem}
\begin{proof}
    Let the first item to arrive online occur with frequency $\alpha \in (0,1)$, i.e., there are $\alpha N$ copies of that type in the input. Let $\mathbb{E}$ be the event that Process \ref{alg:proc-1} returns $1$. We get that:
    \begin{align*}
        \Pr[b = 1] &= \Pr[b = 1 \cap \mathbb{D}] + \Pr[b = 1 \cap \mathbb{I}]\\
        &= \Pr[\sigma(I_2) < \sigma(I_1) \;|\; \mathbb{D}] \cdot \Pr[\mathbb{D}] + \Pr[\mathbb{E} \; | \; \mathbb{I}] \cdot \Pr[\mathbb{I}]\\
        &= \frac{1}{2}(1-\alpha) + \alpha \cdot \Pr[\mathbb{E} \; | \; \mathbb{I}]
    \end{align*}
    where the last equality follows from Lemma \ref{lem:half-distinct}.\\
    Furthermore:
    \begin{align*}
      \Pr[\mathbb{E} \; | \; \mathbb{I}] &= (1-\alpha) + \alpha^2(1-\alpha) + \alpha^4(1-\alpha) +...\\
      &= (1-\alpha)\sum_{i=0}^{\infty} \alpha^{2i}\\
      &= (1-\alpha)\frac{1}{1-\alpha^2} = \frac{1}{1+\alpha}
    \end{align*}
    Putting everything together:
    \[\Pr[b = 1] = \frac{1}{2}(1-\alpha) + \frac{\alpha}{1+\alpha}\]
    which achieves a maximum value of $2-\sqrt{2}\approx 0.585$ at $\alpha = \sqrt{2}-1 \approx 0.414$. Specifically, we have that $\Pr[b = 1] \in (\frac{1}{2},2-\sqrt{2}]$, which concludes the proof.
\end{proof}

Aside from the processes described above, there is an extensive body of work on randomness extractors contributed by both complexity theorists and cryptographers. We are not aware of any existing deterministic extractors that apply to our input source. 

We considered the extractors discussed in Vadhan \cite{vadhan2012pseudorandomness} and why they either exclude random-order input sources or why the corresponding extractor constructions do not apply to our model (aside from their offline nature.) We also note that the extractors typically studied in the literature are designed to extract many bits (often polynomial in the input size), whereas our goal is to extract only a single bit. Despite this weaker requirement, extracting even a single bit \textit{when needed} with provable guarantees on its bias is a nontrivial task. The closest analog to our extractor is the classical ``von Neumann trick'' that considers 2 biased bits at a time and returns an unbiased bit if the two biased bits are not equal. The von Neumann method cannot guarantee an unbiased bit when needed and can only guarantee an unbiased bit in expectation.

The first sources that Vadhan examines are \textit{IID-bit sources} and \textit{independent-bit sources}, in which the input is a sequence of random variables $X_1,\dots,X_N \in \{0,1\}$ that are independent but biased. In IID-bit sources, all bits have the same nonzero bias, whereas in independent-bit sources different bits may have different biases. In our setting, the adversary is free to determine the multiset of items that are eventually permuted, which may range from $N$ distinct items to a single repeated item. The bits are not independent since each observed item is drawn without replacement from the adversarially specified multiset.

The next class that Vadhan discusses is \textit{Santha-Vazirani sources}, which satisfy the property that there exists some constant $\delta >0$ such that for every $i \in [N]$ and every $x_1,\dots,x_{i-1} \in \{0,1\}$, $\delta \leq \Pr[X_i=1|X_{1}=x_{1}, \dots, X_{i-1}=x_{i-1}]\leq 1-\delta$. Random-order inputs do not satisfy this condition. Indeed, an adversary may conjure a multiset consisting of $N-1$ identical items together with a single distinct item. Once we condition on having observed the unique item (or all the identical items), the remaining item(s) are fully determined. That is, there is some $j \in [N]$ for which $\Pr[X_j=1|X_{1}=x_{1}, \dots, X_{j-1}=x_{j-1}] \in \{0,1\}$. Hence, the $j^{th}$ bit may be completely determined by the previous $j-1$ bits (or by the remaining $N-j$ bits.)


Next discussed are \textit{$k$-sources}, which are distributions having min-entropy at least $k$. Inputs in the random-order model may have min-entropy ranging from $\log N$ (when only one item is distinct) to $\log (N!)=\Theta(N\log N)$ (when all items are distinct), and thus constitute a $\log N$-source. In particular, the entropy of the source depends entirely on the adversary's strategy. It is well known that (universal) deterministic extractors for $k$-sources do not exist, though it may be possible to construct extractors for specific sources. We are not aware of any existing extractors in the literature that produce a random bit when our applications require it. While seeded extractors that require a small number of truly random bits exist, even a small amount of randomness would defeat our goal of derandomizing barely random algorithms. 

\section{Applications of the 1-bit extraction method \texttt{COMBINE}}
\label{sec:applications}

We apply the \texttt{COMBINE} procedure to derandomize a number of $1$-bit barely random algorithms. In particular, we derandomize the general and proportional knapsack algorithms of Han et al. \cite{han2015randomized}, the interval selection algorithms of Fung et al. \cite{fung2014improved} for single-length and monotone instances with arbitrary weights as well as C-benevolent and D-benevolent weighted instances, the unweighted throughput scheduling algorithms of Kalyanasundaram and Pruhs \cite{kalyanasundaram2003maximizing} and Chrobak et al. \cite{chrobak2007online}, and the Fung et al. \cite{fung2014improved} algorithm for single-length throughput instances with arbitrary weights.
We also consider derandomizing the algorithm of Albers \cite{Albers02} for the makespan problem on identical machines, but here our algorithm is not as good as the random-order algorithm of Albers and Janke \cite{AlbersJ21}. Finally, we consider the algorithm of Kimbrel and Saia \cite{KimbrelS00} for the two-machine job shop problem, where we obtain no improvement over what is already known for deterministic algorithms under adversarial order.

With the exception of the application to single-length instances, the weighted interval selection results and the unweighted throughput results are in the {\it real-time with random-order model}, where the set of release times in the set of input items are fixed and a random permutation is then applied to the remaining parameters of each input item. We will precisely specify this real-time model for the interval selection and throughput problems. Following most of the literature, we state our results as ratios greater than 1.

\subsection{Knapsack}

\textbf{Input:}\>\> A sequence of items $((v_1,w_1), \dots, (v_n,w_n))$ arriving online. The value $v_i \in \mathbb{R}_{\geq 0}$ and the weight $w_i \in \mathbb{R}_{\geq 0}$ of item $i$ are only revealed once it arrives. The knapsack has capacity $C$ which is known to the algorithm before any item arrives.\\
\textbf{Output:}\>\> A vector $\bm z = (z_1,z_2,\dots,z_n) \in \{0,1\}^n$, where $z_i=1$ if and only if item $i$ is packed after processing the final item. We study the knapsack problem in two models: one with revoking (i.e., the ability to forever discard any previously packed item(s)) and one without. In the model without revoking, an algorithm must decide $z_i$ upon the arrival of item $i$ (and before the arrival of item $i+1$.) Conversely, the model with revoking allows an algorithm to zero $z_i$ anytime after its arrival. \\
\textbf{Objective:}\>\> To find $\bm z$ such that $\sum_{i=1}^nz_iv_i$ is maximized subject to $\sum_{i=1}^nz_iw_i \leq C$.\\

In the online knapsack problem, an algorithm maintains a knapsack of capacity $C$ and receives items sequentially.
The objective criterion is to maximize the total value that can be packed in the knapsack subject to the capacity constraint $C$. In the ROM, an adversary selects a multiset of items $\{(v_1,w_1),\dots,(v_n,w_n)\}$ to which a uniformly random permutation is applied. The items then arrive online in the order of the permutation.

The knapsack problem is a central problem in complexity theory and the analysis of algorithms, and is also well studied in the online literature beginning with Marchetti-Spaccamela and Vercellis \cite{Marchetti-SpaccamelaV95} who provide a non-constant inapproximation for deterministic online algorithms. This was extended by Zhou, Chakrabarty and Lukose \cite{ZhouCL08} who prove a  randomized $\ln(U/L)$ inapproximation for the general knapsack where $U$ (respectively, $L$) is an upper (respectively, lower) bound on the value density of input items. A $2$-competitive randomized algorithm for the proportional knapsack problem is due independently to B{\"{o}}ckenhauer et al. \cite{bockenhauer2014online} and Han et al. \cite{han2015randomized} who prove that $2$ is the optimal randomized
competitive ratio for randomized online algorithms without revoking. With revoking, Han et al. also provide a barely random $\frac{10}{7}$-competitive algorithm for the proportional knapsack problem.  Our ROM result for the general knapsack problem is based on the Han et al. $2$-competitive randomized online algorithm with revoking. Also very relevant to our work is the random-order algorithms (without revoking) for the knapsack problem which is known as the knapsack secretary problem.  Babaioff et al. \cite{BabaioffIKK07} initiated the study of the knapsack secretary problem and developed a $16e$ competitive algorithm.  This was followed by the $8.06$ competitive algorithm of  Kesselheinm et al. \cite{KesselheimRTV18} and the current best $6.65$ randomized ROM competitive algorithm of Albers, Khan and Ladewig \cite{AlbersKhanLadewig2021}. Like the Garg et al. algorithm for interval selection, all these knapsack secretary algorithms are randomized and require knowledge of the number  of input items. The latter two papers extend their results (with somewhat worse ratios) to GAP, the generalized  assignment problem. In contrast, we will derive a much simpler deterministic online  algorithm for the knapsack problem using revoking instead of randomness. 
\subsubsection{General Weights}
Han et al. \cite{han2015randomized} study the online (with revoking) knapsack problem in the adversarial setting and give a $2$-competitive randomized algorithm that initially chooses between two different deterministic algorithms (one that greedily maintains the best items in terms of \textit{value-density} $\frac{v_i}{w_i}$ and one that packs the best items in terms of value) uniformly at random. A key point in derandomizing their algorithm is that as long as identical copies of the first item arrive, the two algorithms behave identically, and when a different item arrives, we extract the random bit. 

 \begin{theorem}
     Derandomizing the general knapsack algorithm (with revoking) of Han et al. \cite{han2015randomized} yields a deterministic $\frac{1}{\sqrt2-1}\approx 2.414$-competitive algorithm for the general knapsack problem with revoking in the ROM.
 \end{theorem}
 \begin{proof}
     We denote the algorithm that packs maximum-value items by \proc{Max}, and the algorithm that packs items with highest value-density by \proc{Greedy}. Let \proc{Opt}$(\sigma)$ denote the value of the optimal offline knapsack, and let \proc{Max}$(\sigma)$ and \proc{Greedy}$(\sigma)$ denote the values of the knapsacks obtained by their respective algorithms on a permuted input instance $\sigma$.
    
     The derandomized algorithm begins by accepting as many identical copies of the first item as can fit. When the first distinct item arrives, we extract a bit with a worst-case bias of $2-\sqrt2$. If the extracted bit is $b=1$, then the algorithm continues with \proc{Greedy}; otherwise, it switches to \proc{Max}.

    Let $\mathbb{E}[$\proc{Alg}$(\sigma)]$ denote the expected value of the knapsack constructed by the derandomized algorithm. Han et al. \cite{han2015randomized} show that $\proc{Greedy}(\sigma)+\proc{Max}(\sigma)\geq \proc{Opt}(\sigma)$. Since the derandomized algorithm randomly selects between the knapsacks produced by \proc{Greedy} or \proc{Max} with probability at least $\sqrt2-1$ each, $\mathbb{E}[\proc{Alg}(\sigma)] \geq (\sqrt2-1)(\proc{Greedy}(\sigma)+\proc{Max}(\sigma)) \geq (\sqrt2-1)\proc{Opt}(\sigma)$.
 \end{proof}
 
 \subsubsection{Proportional Weights}
 The proportional knapsack problem is the special case of the knapsack problem where $w_i=v_i$. We consider two models -- one with revoking and another without. In the model without revoking, a simple $2$-competitive 1-bit barely random algorithm maintains two buckets, each with capacity $C$. The algorithm attempts to fill the first bucket and places into the second bucket items too large to fit into the first one. Items that cannot fit into either bucket are discarded. It selects between the two buckets with equal probability.

Generally, any strategy that begins by greedily accepting some number of identical items will require the ability to revoke to be $O(1)$-competitive. We prove this in Appendix \ref{appendix:revocations-required}, and also demonstrate that this strategy (with revoking permitted) is $\frac{1}{\sqrt2-1}$-competitive. Unfortunately, $\frac{1+\sqrt5}{2} \approx 1.618$ is the tight deterministic bound in the adversarial model with revoking. Hence, this derandomization (even with revoking power) fails to make good use of the random-order assumption.

We next derandomize the $\frac{10}{7}$-competitive algorithm of Han et al. \cite{han2015randomized} for the proportional knapsack problem with revoking. The algorithm assigns to each item a weight class and maintains two deterministic algorithms that have different acceptance criteria for different classes of weights. Our derandomization strategy is as follows: before extracting a bit, pack as many identical copies of the first item as possible. Then, pivot to the algorithm of Han et al. This derandomization is $\frac{5}{3\sqrt2-1}\approx 1.542$-competitive, with only a slight loss in the ratio from the bias in the bit.

\begin{theorem}
     Derandomizing the proportional knapsack algorithm (with revoking) of Han et al. \cite{han2015randomized} yields a deterministic $\frac{5}{3\sqrt2-1}\approx 1.542$-competitive algorithm for the proportional knapsack problem with revoking in the ROM.
\end{theorem}
\begin{proof}
    See Appendix \ref{appendix:han-prop-knapsack} for the analysis of the derandomization. In particular, the derandomized algorithm is $\frac{5}{3\sqrt2-1}\approx 1.542$-competitive which improves upon the deterministic impossibility bound of $\frac{1+\sqrt5}{2} \approx 1.618$ in the adversarial model. 
\end{proof}

\subsection{Interval Selection}\label{sec:intervals}
\textbf{Input:}\>\> A sequence of intervals $\mathcal{J} = (J_1, \dots, J_n)$ arriving online. Each interval is represented by a triple $(r_i,p_i,w_i) \in \mathbb{R}_{\geq 0}^3$ where $r_i$ is the release time when the interval starts, $r_i + p_i$ is the time when the interval finishes, and $w_i$ is the weight of the interval. We assume $p_i,w_i>0$ for all $i$ and we assume that an interval can start when another interval finishes. Let $d_i=r_i+p_i$ denote the finishing time of the interval. \\
\textbf{Output:}\>\> A feasible set $S$ of intervals in which intervals do not intersect.\\
\textbf{Objective:}\>\> To find a feasible set that maximizes the total sum of weights $\sum_{i \in S} w_i$ of scheduled intervals.




Interval selection is arguably the most studied special case of the throughput problem (in both the online and offline settings), which we consider in the next section. We refer the reader to the surveys by Kolen et al. \cite{kolen2007interval} and Kovalyov et al. \cite{kovalyov2007fixed} for references and the many applications of interval scheduling.  In the unweighted (respectively, weighted) setting, the goal is to maximize the cardinality (respectively, the sum of the weights)  of the selected subset of intervals. 
 Lipton and Tomkins \cite{lipton1994online} introduced the problem of real-time interval selection. Their paper also introduced the classify and randomly select paradigm. In the traditional online version of the problem, intervals arrive one at a time and the algorithm must either permanently accept an interval or forever discard it.
 Results in this worst case adversarial setting are quite negative. Even in the real-time model where intervals arrive in order of increasing starting times, Bachmann et al. \cite{bachmann2013online} establish an $\Omega(n)$ lower bound for randomized algorithms. We consider the model where a new interval can be accepted, displacing any conflicting intervals currently in the solution. Similar to intervals that are rejected upon arrival, displaced intervals  can never be taken again. In the real-time unweighted setting, Faigle and Nawijn \cite{faigle1995note} consider a simple greedy $1$-competitive deterministic algorithm with revoking. In the same real-time with revoking setting, Woeginger \cite{woeginger1994line} establishes multiple 4-competitive deterministic algorithms for special classes of weight functions that we will discuss. For arbitrary weights, Woeginger (respectively, Canetti and Irani \cite{canetti1995bounding}) show that even in the real-time setting with revoking, there cannot be a constant deterministic (respectively, randomized) competitive ratio. We will consider special classes of weight functions from the work of Fung et al. \cite{fung2014improved} for which they establish randomized 2-competitive algorithms. We describe the derandomization for each of these Fung et al. algorithms in Appendix \ref{appendix:fung} and refer the reader to \cite{fung2014improved} for a complete description of each. 



In interval selection, we can assume the online or the real-time model.  Since interval selection is a special case of the throughput problem, it is often studied in the real-time model. Although it seems inconsistent to simultaneously consider real time arrivals and random order arrivals, there is a meaningful sense in which we can consider {\it real-time intervals in the random-order model}. In the random-order model, we leave $r_i$ fixed and and apply a permutation $\pi$ to obtain the $i^{th}$ interval $(r_i, p_{\pi(i)},w_{\pi(i)})$. In the real-time model the starting points of intervals are fixed, but the remaining attributes of an interval (e.g., processing time and weight) are randomly permuted. Distinct intervals are lexicographically ordered, and two intervals $i$ and $j$ are \textit{pseudo-identical} if $p_i=p_j$ and $w_i=w_j$. For C-Benevolent and D-Benevolent instances, $p_i=p_j$ implies $w_i = w_j$ since the weight of each interval is a fixed function even if the release times are different\footnote{Specifically, an instance is C-benevolent if each weight $w_i=f(p_i)$ for a strictly increasing, continuous function $f$ satisfying $f(0)=0$, $f(p)>0$ for all $p>0$, and $f(\ell_1)+f(\ell_2) \leq f(\ell_1-\epsilon)+f(\ell_2+\epsilon)$ for all $0<\epsilon\leq \ell_1 \leq \ell_2$. An instance is D-benevolent if each weight $w_i=f(p_i)$ for a function $f$ satisfying $f(0)=0$, $f(p)>0$ for all $p>0$ that is non-increasing on $\mathbb{R}^+$.}.  
We consider the single-length, monotone, C-benevolent and D-benevolent interval selection algorithms of Fung et al. \cite{fung2014improved} and demonstrate that each can be derandomized in the ROM. In each instance, our derandomization is $\frac{1}{\sqrt2-1}\approx2.414$-competitive.


\subsubsection{Single-length instances with arbitrary weights}
\label{subsubsection:single-length-derandomization}

For single-length instances, the length (or processing time) is fixed; that is $p_i = p$ for all $i$. We note that for single-length instances, the Fung et al. algorithm applies to the online model (i.e., without the real-time assumption) while all the other interval selection and throughput problems are with respect to the real-time model. We begin by describing the 2-competitive barely random algorithm of Fung et al. \cite{fung2014improved} for single-length intervals with arbitrary weights, and then describe how it generalizes to monotone, C-benevolent and D-benevolent instances.

Fung et al. define \textit{slots} as fixed intervals on the real line. Suppose that the first interval is released at time 0. Slots $s_1, s_2, \dots$ are defined as intervals $[0,p), [p,2p), \dots$ where $p$ is the processing requirement of each interval. 
One deterministic algorithm schedules the heaviest jobs released in odd slots $s_1, s_3, \dots$, while the other algorithm schedules the heaviest jobs released in even slots $s_2, s_4, \dots$. Note that an interval started in $s_i$ may execute into $s_{i+1}$, but never extends into $s_{i+2}$ since processing times are identical. By randomly choosing to take intervals from either odd or even slots, we guarantee a feasible solution that obtains at least one-half the weight of an optimal solution. 

In Theorem \ref{thm:fung-single-length-derand}, we derandomize the algorithm of Fung et al. \cite{fung2014improved} to obtain a $\frac{1}{\sqrt2-1}\approx2.414$-competitive algorithm in the ROM. Our derandomization handles pseudo-identical items that arrive before a random bit can be extracted and then transitions to executing the real-time interval selection algorithm of \cite{fung2014improved}. We also describe derandomizations for the remaining interval selection algorithms (for monotone, C-benevolent, and D-benevolent instances) in Appendix \ref{appendix:fung}.

\begin{theorem}\label{thm:fung-single-length-derand}
     Derandomizing the single-length interval selection (with revoking) algorithm of Fung et al. \cite{fung2014improved} yields a deterministic $\frac{1}{\sqrt2-1}\approx2.414$-competitive algorithm for the single-length interval selection problem (with revoking) in the ROM.
\end{theorem}
\begin{proof}
 The derandomized algorithm begins by greedily scheduling as many non-overlapping pseudo-identical intervals as possible until an interval with a distinct weight arrives. Suppose that the first interval with a distinct weight arrives at time $r_{i+k}$, where interval $i$ is the last interval scheduled greedily. Define the first slot $s_1$ to begin at $r_{i}$, and then execute the algorithm of Fung et al. \cite{fung2014improved}, replacing their fair coin flip with the extracted bit. Intervals $j \in \{i+1,\dots,i+k-1\}$ are released during $s_1$ and have identical weight $w_i$, so discarding these intervals does not affect optimality.

 Since the extracted bit has a worst-case bias of $2-\sqrt2$, the derandomized algorithm selects either slot parity with probability at least $\sqrt2-1$. Note that the resulting schedule is non-conflicting: the selected intervals in $\{1, \dots, i-1\}$ do not conflict because they are identical and are scheduled with our greedy non-overlapping rule, and the selected intervals in $\{i, \dots, n\}$ are chosen by the algorithm of Fung et al.

Let $\proc{Alg}_{\leq i-1}(\sigma)$ denote the total weight of intervals in $\{1, \dots, i-1\}$ scheduled by our algorithm, and let $\proc{Opt}_{\leq i-1}(\sigma)$ denote the optimal offline weight over these intervals. Let $\proc{Opt}_{\geq i}(\sigma)$ denote the optimal offline weight over intervals in $\{i, \dots, n\}$, and let $\mathbb{E}[\proc{Alg}(\sigma)]$ denote the expected total weight of intervals scheduled by the derandomized algorithm. Let $\proc{Opt}(\sigma)$ denote the optimal offline weight over all intervals. Further, for each slot $s_j$, let $W_j$ denote the weight of the heaviest interval released during that slot. We note the following observations:

\begin{enumerate}
    \item $\proc{Alg}_{\leq i-1}(\sigma) = \proc{Opt}_{\leq i-1}(\sigma)$
    \item $\proc{Opt}(\sigma) \leq \proc{Opt}_{\leq i-1}(\sigma)+\proc{Opt}_{\geq i}(\sigma)$
    \item $\proc{Opt}_{\geq i}(\sigma) \leq \sum_{\text{odd } j}W_j+\sum_{\text{even } j}W_j$
\end{enumerate}

Observation (1) follows since the deadlines of the pseudo-identical intervals are non-decreasing and selecting non-conflicting intervals by earliest-deadline-first is optimal offline. Observation (2) follows since the right-hand side of the inequality corresponds to a relaxation that allows overlapping intervals. Observation (3) follows since at most one interval can be scheduled in each slot. Since slots are defined starting from time $r_i$, the algorithm of Fung et al. considers intervals $\{i, \dots, n\}$. Hence,

\begin{align*}
    \mathbb{E}[\proc{Alg}(\sigma)] &\geq (\sqrt2-1)(\proc{Alg}_{\leq i-1}(\sigma)+\sum_{\text{odd $j$}}W_j)+(\sqrt2-1)(\proc{Alg}_{\leq i-1}(\sigma)+\sum_{\text{even $j$}}W_j) \\
    &= (\sqrt2-1)(2\proc{Alg}_{\leq i-1}(\sigma)+\sum_{\text{odd $j$}}W_j+\sum_{\text{even $j$}}W_j) \\
    &\geq (\sqrt2-1)(\proc{Opt}_{\leq i-1}(\sigma)+\proc{Opt}_{\geq i}(\sigma)) \\
    &\geq (\sqrt2-1)\proc{Opt}(\sigma)
\end{align*}

The second inequality follows from Observations (1) and (3), and the third inequality follows from Observation (2).
Hence, the derandomized algorithm is $\frac{1}{\sqrt2-1}\approx2.414$-competitive.
\end{proof}

\subsubsection{Monotone, C-Benevolent and D-Benevolent Weight Functions}
\label{sec:mono-cbenev-dbenev-intervals}

We again wish to maximize the weighted sum of intervals accepted by a single machine. An instance is \textit{monotone} if for every pair of intervals $i$ and $j$, if $r_i<r_j$ then $d_i \leq d_j$. In the context of interval selection, an instance is \textit{C-benevolent} (respectively, \textit{D-benevolent}) if the weight of an interval is a strictly increasing convex function of its length (respectively, a monotonically decreasing function of its length.) Note that for C-benevolent and D-benevolent instances, the processing time determines an interval's weight. Note also that D-benevolent functions include the constant function $f(p) = 1$, so results for D-benevolent functions apply to the unweighted case. 

 Fung et al. \cite{fung2014improved} give 1-bit barely random algorithms that adapt their single-length interval selection algorithm for monotone, C-benevolent, and D-benevolent instances of interval selection in the real-time model. In these settings, even and odd slots are determined at runtime so that at least one of the two deterministic algorithms is eligible to start an interval at any time in each slot. As with single-length instances, we randomly choose between either even or odd slots, but now have to be more careful about an initial sequence of pseudo-identical intervals. These algorithms are 2-competitive and we obtain the following result with a derandomization similar to that for single-length instances. 
 \begin{theorem}
      Derandomizing the interval selection (with revoking) algorithms of Fung et al. \cite{fung2014improved} for monotone instances, as well as instances with C-benevolent and D-benevolent weight functions yields deterministic $\frac{1}{\sqrt2-1}\approx2.414$-competitive algorithms for each corresponding variant of the interval selection problem (with revoking) in the ROM.
 \end{theorem}
 \begin{proof}
      We sketch a derandomization that applies to all three instances in Appendix \ref{appendix:fung}. The proof for each derandomization is structurally similar to the proof of Theorem \ref{thm:fung-single-length-derand}.
 \end{proof}
 

\subsection{Job Throughput Maximization}
\textbf{Single-Machine Throughput Maximization}\\
\textbf{Input:}\>\> A sequence of jobs $\mathcal{J} = (J_1, \dots, J_n)$ arriving online. Each job is represented by a tuple $J_i = (r_i,p_i,d_i,w_i) \in \mathbb{R}^4_{\geq 0}$, where $r_i$ is the release time, $p_i$ is the processing time, $d_i$ is the deadline, and $w_i$ is the weight of the job. We have $d_i \geq r_i$ for all $i$, and $r_1 \leq \dots \leq r_n$. Equivalently, each job can be represented as $(r_i,p_i,s_i,w_i)$ where $s_i = d_i-p_i-r_i$ denotes its slack.\\
\textbf{Output:}\>\> A feasible schedule of jobs on the machine in which every scheduled job begins only after it is released, each completed job $J_i$ executes for at least $p_i$ time units, and no two jobs execute simultaneously.\\
\textbf{Objective:}\>\> To find a feasible schedule that maximizes the total weight of scheduled jobs that complete before their deadlines.

The unweighted job throughput problem is the special case of the job throughput problem for when $w_i = 1$ for all jobs.
Like interval selection, we consider a modified model of random arrivals where the starting points of jobs are fixed, but the remaining attributes of a job (e.g. processing time and slack) are randomly permuted. That is, the $i^{th}$ job to arrive is given by
$(r_i,p_{\pi(i)},s_{\pi(i)},1)$ for a uniformly random permutation $\pi$. The deadline $r_i+p_{\pi(i)}+s_{\pi(i)}$ of the $i^{th}$ job is appropriately adjusted via its slack.

Kalyanasundaram and Pruhs \cite{kalyanasundaram2003maximizing} consider the single processor setting and give a barely random algorithm with a constant (although impractically large) competitiveness for the unweighted real-time throughput problem in the preemption model (with resumption). Baruah et al. \cite{baruah1994line} show that no deterministic real-time algorithm with preemption can achieve constant competitiveness. In Appendix \ref{appendix:KP}, we describe a derandomization of \cite{kalyanasundaram2003maximizing} that is $O(1)$-competitive in the ROM.
\begin{theorem}
     Derandomizing the unweighted throughput (with preemption) algorithm of Kalyanasundaram and Pruhs \cite{kalyanasundaram2003maximizing} yields a deterministic $\approx 311491.491$-competitive algorithm for the unweighted throughput problem (with preemption) in the ROM.
\end{theorem}
\begin{proof}
    See Appendix \ref{appendix:KP}.
\end{proof}

In the model that allows restarting of jobs, Hoogenveen, Potts and Woeginger \cite{HoogeveenPW00} provide a 2-competitive deterministic algorithm for the unweighted throughput problem. He \cite{He2025} shows that in the revoking model, when the slack $s_i$ is large compared to the processing time (namely, $s_i \geq 2p_i$) for each job, there is no constant competitive real-time algorithm for unweighted throughout. 
Chrobak et al. \cite{chrobak2007online} consider a special case of the unweighted throughput problem where all jobs have equal processing time. Their $1$-bit barely random algorithm runs two copies of the same deterministic algorithm that synchronize via a shared lock, leading to asymmetrical behavior and achieving a competitive ratio of $\frac{5}{3}$.
We derandomize their algorithm and only incur a small constant penalty to the competitive ratio due to the bias in the extracted bit; if the bit is unbiased, then we exactly match their $\frac{5}{3}$-competitiveness under random arrivals.
\begin{theorem}
     Derandomizing the equal-length unweighted throughput algorithm of Chrobak et al. \cite{chrobak2007online} yields a deterministic $\frac{5}{2\sqrt2}\leq 1.768$-competitive algorithm for the equal-length unweighted throughput problem in the ROM.
\end{theorem}
\begin{proof}
    The (non-trivial) proof is given in Appendix \ref{appendix:chrobak}. 
\end{proof}

In addition to interval selection, Fung et al. \cite{fung2014improved} consider the single-machine job throughput problem with restarting in which jobs have equal lengths and arbitrary weights. They give a $3$-competitive randomized algorithm that similarly defines slots on the real line and executes two algorithms that each selects the heaviest jobs that can start within slots of a fixed parity. A simple derandomization of their algorithm in the ROM model beats the deterministic lower bound of $4$ in the adversarial model established by Woeginger \cite{woeginger1994line}.

\begin{theorem}
     Derandomizing the equal-length job throughput (with restarting) algorithm of Fung et al. \cite{fung2014improved} yields a deterministic $\frac{3}{2\sqrt2-2}\approx3.621$-competitive algorithm for the equal-length job throughput problem (with restarting) in the ROM.    
\end{theorem}
\begin{proof}
    See Appendix \ref{appendix:fung}.
\end{proof}

\subsection{Makespan Minimization}

\subsubsection{Two-Machine Job Shop}
\textbf{Two-Machine Job Shop Scheduling}\\
\textbf{Input:}\>\> 
All jobs are released at time $0$. Each job consists of a sequence of operations where each operation makes a request for service on exactly one of the two machines. Each operation has a processing requirement whose value is unknown until the operation completes, and subsequent operations are revealed online only after their predecessors complete. The operations of each job alternate between the two machines so that no job requests the same machine in consecutive operations. Initially, only the first operation of each job is revealed to the algorithm.\\
\textbf{Output:}\>\> A (possibly preemptive) feasible schedule of operations on the machines in which each operation begins only after its predecessor completes, runs on its requested machine, and no two operations execute simultaneously on the same machine.\\
\textbf{Objective:}\>\> To find a feasible schedule that minimizes the makespan. \\

In the general job shop environment, we have some number of machines and jobs, each with a sequence of operations that must be processed on specific machines. In the ROM, the processing times of operations are set by the adversary and a uniformly random permutation is applied to these times. As operations are scheduled online, each is assigned the next processing time in the permutation. In this model, an operation may be assigned the processing time of an operation originally belonging to a different job. The remaining constraints (i.e., the machine requirements and the number of operations belonging to each job) are adversarially determined and hence cannot be used for randomness extraction. To the best of our knowledge, no prior work has studied the two-machine job shop problem in the random-order model. 

We motivate this random-order application with the following example: consider a factory that produces different sorts of toys. Each toy passes through the factory in stages that alternate between assembly and quality control. Some toys are assembled entirely in-house, while others receive partially completed components from intermediary suppliers and begin at the quality control stage. Different toys require different amounts of time for assembly and quality control; some pass inspection immediately while some may require extensive rework and similarly assembly times may vary depending on the worker processing the toy. The time required for each stage is unknown until it completes. In what order should toys (and their components) be processed to minimize the time required to complete all toys?

Kimbrel and Saia \cite{KimbrelS00} consider a two-machine job shop that allows preemption, i.e., a machine may interrupt the running operation in favor of executing an operation of another job. The interrupted operation may be resumed later. The optimality criterion is to minimize the makespan. Kimbrel and Saia present a $\frac{3}{2}$-competitive 1-bit barely random algorithm that begins by arbitrarily assigning a priority order to the jobs. One deterministic algorithm runs jobs in priority order on machine 1 and in reverse priority order on machine 2, whereas the other runs jobs in reverse priority order on machine 1 and in priority order on machine 2. Derandomization becomes difficult if the random bit forces the machines to reverse their processing orders.
\begin{theorem}
     Derandomizing the online two-machine job shop (with resuming) algorithm of Kimbrel and Saia \cite{KimbrelS00} yields a deterministic $10-5\sqrt2\leq 2.929$-competitive algorithm for the online two-machine job shop problem (with resuming) in the ROM.
\end{theorem}
\begin{proof}
See Appendix \ref{appendix:KimbrelSaia}.
\end{proof}

Here, our loss in the competitive ratio is due both to the bias in the bit and to the cost of reversing the processing order. However, the tight deterministic ratio is $2$ in the adversarial model \cite{KimbrelS00}.

\subsubsection{$m$ Identical Machines}

\textbf{Makespan Minimization on $m$ Identical Machines}\\
\textbf{Input:}\>\> A sequence of jobs $(p_1, \dots, p_n)$ arriving online, where $p_i >0$ is the processing requirement for job $i$. The number of identical machines $m$ is known to the scheduler before any jobs arrive.\\
\textbf{Output:}\>\> An assignment $\sigma : [n]  \to [m]$ where $\sigma(i)=j$ denotes the assignment of job $i$ to machine $j$. $\sigma(i)$ must be decided upon the arrival of job $i$ and cannot be changed thereafter.\\
\textbf{Objective:}\>\> To find an assignment $\sigma$ that minimizes $\max\limits_j \sum\limits_{i:\sigma(i)=j}p_i$. \\

In the ROM, an adversary selects a multiset of jobs with processing times $\{p_1, \dots, p_n\}$ to which a uniformly random permutation is applied. The jobs then arrive online in the order of the permutation. 

Albers \cite{Albers02} provides a $1.916$-competitive 1-bit barely random algorithm for the makespan problem on identical machines. The algorithm initially chooses between one of two scheduling strategies that differ in the fraction of machines that they try to keep lightly and heavily loaded. One strategy tracks the decisions of the other in order to correct for any mistakes it might make. Because these deterministic strategies behave differently on a stream of identical inputs, derandomizing this algorithm under random arrivals may incur an additional constant penalty to the competitive ratio in the ROM (beyond loss from the bias.) 

Albers and Janke \cite{AlbersJ21} study the identical-machines makespan problem in the ROM and devise a $1.848$-competitive deterministic algorithm, which our derandomization cannot beat. By relaxing the problem to allow unlimited recourse (i.e., the ability to reassign arbitrarily many jobs)\footnote{This is essentially the offline makespan problem.} we can apply our technique to select one of the two strategies and rearrange jobs to match the current state that the selected strategy would have been in, had it been used to schedule every job. Then, the ROM algorithm will differ only in the probability of selecting a strategy. As a result, it achieves a competitive ratio close to Albers' $1.916$ (some bias may improve the performance of the algorithm, but only up to $1.915$ \cite{Albers02}). That is, our derandomization cannot hope to improve upon the best deterministic upper bound of $1.848$ in the ROM. We provide this derandomization to keep consistent our claim that any 1-bit algorithm that we are aware of can be derandomized with our technique under random-order arrivals. 

If we remain faithful to the original problem and instead assume no recourse then some greedy rule must be used to schedule the initial identical jobs. Unfortunately, the deterministic algorithms of Albers behave very differently even on identical jobs. Therefore, it is implausible to wait for a bit to be extracted to switch between the two without some notion of recourse (which we do not allow.) The competitive ratio of our derandomization suffers from our initial greedy decisions.

\begin{theorem}
     Derandomizing the online identical-machines makespan algorithm of Albers \cite{Albers02} yields a deterministic $1.916(4-2\sqrt2)+1\approx 3.245$-competitive algorithm for the online identical-machines makespan problem in the ROM.
\end{theorem}
\begin{proof}
    See Appendix \ref{appendix:albersmakespan}.
\end{proof}

\section{Going Beyond 1 bit of Randomness}\label{sec:beyond-1-bit}

While our extraction processes are so far limited to a single bit, there are some ideas that can be used for extracting some small number of random bits depending on the application and what we can assume. In the online (i.e. any-order) model, we know from \cite{borodin2023any} that we can extend the Fung et al. algorithm for single-length instances to obtain a classify and randomly select $2k$-competitive algorithm for instances with at most $k$ different interval lengths and arbitrary weights. Furthermore, $k$ does not have to be known initially, and we can randomly choose a new length whenever it first appears. In the random-order model, we can derandomize that algorithm for the special case of two-length interval selection with arbitrary weights under the assumption that the input sequence consists of distinct intervals. However, we note that this distinct items assumption may not be without loss of generality for the algorithms we consider for interval selection.
This way we are able to utilize both Process \ref{alg:proc-1} and \ref{alg:proc-2} and get a $6$-competitive deterministic algorithm. The algorithm would use the unbiased bit from the first two intervals to decide on the length. While working on any length, we would use Process \ref{alg:proc-1} to decide on the slot type.

The above algorithm would not necessarily work in the case of three lengths, because if the first three intervals are of different length, we would not have extracted enough random bits to simulate the $\frac{1}{3}$ coin required for the third length. Alternatively, since we can assume distinct interval items, we are able to use the relative order of the second and third (distinct) items noting that this bit is correlated with the first two intervals but still yields a bit with a $\frac{2}{3}$ bias. We can extend this idea of utilizing the relative order of adjacent distinct items to derive deterministic ROM algorithms for any number of interval lengths $k$ (noting that we would have to approximate probabilities that are not equal to $\frac{1}{2^j}$ for some $j$). The competitive ratio would deteriorate rapidly with increasing $k$.

Also note that our extraction technique is not tied to the ROM and extends naturally to related settings. Recalling that a ROM algorithm implies the same competitive ratio for input items coming from an (even unknown) i.i.d. source, we note that if the source is uniformly distributed over an infinite or sufficiently large support, then with high probability every new input item generates an unbiased bit. Hence, we can derive a deterministic algorithm for weighted interval selection for an i.i.d. source uniformly distributed over a sufficiently large support. Of course, as we have noted the approximation will substantially deteriorate when the number of  distinct interval lengths increases.

Again, assuming that we have distinct items, we can easily extract two independent bits by comparing the first and second items, and then the third and fourth items. Unlike the previous application where we may need the second random bit as soon as three items have arrived, in unweighted problems we can afford to not worry about losing a few items. The Emek et al. \cite{emek2016space} online algorithm\footnote{Which operates in the online (i.e., any-order) model, in contrast to the real-time algorithms of Fung et al. \cite{fung2014improved}.} with revoking for unweighted interval selection obtains a 6-competitive algorithm when we randomly chose (with probability $\frac{1}{3}$ for each color) between the colors of a 3-colorable graph. We can waste a fourth color to provide an $8$-competitive ratio and then use two bits to decide between four colors. This shows that a constant competitive ROM result is possible. Recently, Borodin and Karavasilis \cite{BorodinK23} show that the arguably simplest algorithm obtains a 2.5 competitive ROM ratio. While the statement of that algorithm is indeed simple, the algorithm requires some detailed analysis. Our algorithm based on randomness extraction would have established that there is a constant ratio.

\section{Conclusions}
We studied three processes for extracting random bits from uniformly random arrivals, and showed direct applications in derandomizing algorithms for a variety of problems. Our main application results are so far mainly limited to extracting a single random bit. We also considered interval selection (with arbitrary weights) when there are at most two (or $k$) different interval lengths, assuming that all intervals in the input sequence are distinct.  With that assumption, we are able to extract two (or any small number of) biased bits. 

The obvious open question is for what other applications can we obtain constant ratios with and without assuming distinct input items? Are there other applications which only gradually need random bits and not all random  bits initially as often seems necessary in classify and randomly select algorithms? Although the Fung et al. algorithm for single-length instances can be viewed as a classify and randomly select algorithm, we know from \cite{borodin2023any} that random bits do not have to be initially known for instances with a limited number of interval lengths.

 We repeat the question that motivated our interest in online randomness extraction. Namely, is there a natural or contrived problem for which we can provably (or even plausibly) show a better randomized competitive algorithm with adversarial order than what is achievable by deterministic algorithms with random-order arrivals?



\bibliography{bibliography}


\appendix

\newpage

\renewcommand \thepart{}
\renewcommand \partname{}

\part{Appendix}
\parttoc

\newpage

\section{Derandomizing Real-Time Algorithms in the Random-Order Model}

\subsection{Derandomizing the Weighted Interval and Job Selection Algorithms of Fung et al. \cite{fung2014improved}}
\label{appendix:fung}





Fung et al. also considered three special cases of weighted interval selection in the real-time model (with revoking.) Namely, they consider  monotone arrivals with arbitrary weights, as well as instances with C-benevolent or D-benevolent weight functions (see Section \ref{sec:intervals} for formal definitions.) For each case, they present a 1-bit barely random algorithm in which interval acceptance is determined by {\it adaptively} defined  ``slots''. 

We define a phase of the algorithm's execution as the times during which some interval is executing (in the language of scheduling.) A new phase begins whenever the currently executing interval completes before the next interval arrives. Since phases do not intersect, it suffices to consider the competitive ratio in any single phase.

We discuss instances with C-benevolent weight functions; the argument for monotone and D-benevolent instances is similar. At a high level, there are two algorithms $A$ and $B$ that act on alternating slots. Slots are defined adaptively as follows.

Suppose a phase begins with some interval arriving at time $r_0$. If multiple intervals arrive at $r_0$, one algorithm (call it $B$) accepts the longest such interval, which we denote by $I_0 = [r_0,d_0)$. Algorithm $A$ accepts the longest (and hence most valuable) interval $I_1$ that arrives during $(r_0,d_0)$ and completes at time $d_1>d_0$. If there is no such interval $I_1$, the phase ends. When an arriving interval ends at some time $d_1$, a new slot $s_2$ is defined as $[d_0, d_1)$. After $I_0$ finishes, algorithm $B$ will accept the longest interval $I_2$ arriving in $s_2$ and finishing at some time $d_2 > d_1$. Again, if $I_2$ does not exist then the phase ends. Similarly,  slot $s_3$ is defined as $[d_1, d_2)$. 

More generally, for $i >1$, we define slot $s_i = [d_{i-2},d_{i-1})$. The analysis of each phase considers its slots $s_1, s_2, \ldots $. If $i$ is odd, $B$ accepts $I_{i-1}$ while algorithm $A$ accepts the longest interval $I_i$ that arrives in $s_i$ and completes after $d_{i-2}$. For even $i$, the roles of $A$ and $B$ are exchanged. Fung et al. use a random bit to choose between even-numbered slots $s_{2k}$ and odd-numbered slots $s_{2k+1}$. Note that intervals in non-consecutive slots do not interfere with each other.


\begin{theorem}
      Derandomizations of the interval selection (with revoking) algorithms of Fung et al. \cite{fung2014improved} for monotone instances, as well as instances with C-benevolent and D-benevolent weight functions yield deterministic $\frac{1}{\sqrt2-1}\approx2.414$-competitive algorithms for each corresponding variant of the interval selection problem (with revoking) in the ROM.    
\end{theorem}

\begin{proof}
For monotone, C-benevolent and D-benevolent instances, pseudo-identical intervals have equal processing times and weights. As in the derandomized algorithm for single-length intervals with arbitrary weights, we begin by greedily scheduling as many non-overlapping pseudo-identical intervals as possible until an interval with a distinct processing time or weight arrives. Suppose that the first distinct interval arrives at time $r_{i+k}$, where interval $i$ is the last interval scheduled greedily. Define the first slot $s_1$ to begin at $r_i$, and then execute the algorithm of Fung et al., replacing their fair coin flip with the extracted bit.

Let $\proc {Alg}_{\leq i-1}(\sigma)$, $\proc{Opt}_{\leq i-1}(\sigma)$, $\proc{Opt}_{\geq i}(\sigma)$, $\mathbb{E}[\proc{Alg}(\sigma)]$, and $\proc{Opt}(\sigma)$ be defined as in the proof of Theorem \ref{thm:fung-single-length-derand}. To obtain the same $\frac{1}{\sqrt2-1}\approx 2.414$ competitive ratio with the same analysis, it is sufficient to show that:
\begin{enumerate}
    \item $\proc{Alg}_{\leq i-1}(\sigma) =\proc{Opt}_{\leq i-1}(\sigma)$;
    \item $\proc{Opt}(\sigma) \leq \proc{Opt}_{\leq i-1}(\sigma)+ \proc{Opt}_{\geq i}(\sigma)$;
    \item $\proc{Opt}_{\geq i}(\sigma)$ is upper bounded by the total weight of intervals in $\{i, \dots, n\}$ selected by $A$ or $B$.
\end{enumerate}

Observations (1) and (2) follow directly from the proof of Theorem \ref{thm:fung-single-length-derand}, while Observation (3) follows from the $2$-competitiveness of the algorithms of Fung et al. The extracted bit selects algorithm $A$ or $B$ with probability at least $\sqrt2-1$ each, yielding the same competitive ratio as in Theorem \ref{thm:fung-single-length-derand}.
\end{proof}

The reader may observe some similarity between the algorithm of Fung et al. and the predecessor chains of Woeginger \cite{woeginger1994line}. We refer the reader to \cite{fung2014improved} for the way in which Fung et al. define slots for monotone and D-benevolent instances.

Fung et al. extend their barely random algorithm for selecting equal-length intervals with arbitrary weights (with revoking) to the job throughput problem for equal-length jobs with arbitrary weights (with restarting.) They present a $3$-competitive 1-bit barely random algorithm that defines slots the same way as their equal-length interval selection algorithm. Both deterministic algorithms for equal-length jobs schedule the heaviest feasible jobs in each slot of the selected parity, aborting lighter jobs that may be running. At the beginning of execution, the randomized algorithm chooses between the two deterministic algorithms uniformly at random. 

We apply our bit extraction technique to the equal-length job throughput algorithm of Fung et al. With an analysis similar to the case of equal-length intervals with arbitrary weights, we obtain a derandomized algorithm that is $\frac{3}{2\sqrt2-2}\approx 3.621$-competitive in the ROM.

\begin{theorem}
     Derandomizing the equal-length job throughput (with restarting) algorithm of Fung et al. \cite{fung2014improved} yields a deterministic $\frac{3}{2\sqrt2-2}\approx3.621$-competitive algorithm for the equal-length job throughput problem (with restarting) in the ROM.    
\end{theorem}
\begin{proof}
The derandomized algorithm schedules non-overlapping pseudo-identical jobs by earliest-deadline-first (equivalently, by earliest release time.) When a distinct job arrives, the algorithm extracts a bit with worst-case bias $2-\sqrt2$ and uses it to select between even or odd slots. Suppose that job $i$ is the last job scheduled greedily, and let $C_i$ denote its completion time (with $C_i=0$ if no jobs have completed.) We define the first slot $s_1$ to begin at time $C_i$, and then proceed with the $3$-competitive algorithm of Fung et al. 

Let $\mathbb{E}[\proc{Alg}(\sigma)]$ denote the expected total weight of jobs scheduled by the derandomized algorithm, and let $\proc {Alg}_{< C_i}(\sigma)$ denote the weight of jobs completed before time $C_i$ by the derandomized algorithm.

First, observe that there exists an optimal schedule that, before $C_i$, schedules the same jobs at the same times as in the schedule of the derandomized algorithm. This follows from an elementary exchange argument on pseudo-identical jobs released before time $C_i$; any job scheduled by the derandomized algorithm prior to $C_i$ is pseudo-identical and has the earliest deadline among all remaining pseudo-identical jobs. Hence, if $\proc{Opt}(\sigma)$ denotes the total optimal offline weight over all jobs in that schedule and $\proc{Opt}_{<C_i}(\sigma)$ denotes the weight of jobs completing before $C_i$ in that schedule, then $\proc{Alg}_{<C_i}(\sigma)=\proc{Opt}_{<C_i}(\sigma)$.



Finally, let $\proc{Opt}_{\geq C_i}(\sigma)=\proc{Opt}(\sigma)-\proc{Opt}_{<C_i}(\sigma)$ denote the weight of the remaining jobs (which are either pending at time $C_i$, or released after $C_i$) in the optimal schedule.

Since the schedules of the derandomized algorithm and the optimal algorithm match up to time $C_i$, both schedules must have the same pending jobs at time $C_i$. If we update the release times of each of these pending jobs to $C_i$, then $\proc{Opt}_{\geq C_i}(\sigma) \leq \frac{3}{2}(\sum_{\text{odd } j}W_j+\sum_{\text{even }j}W_j)$, where $W_j$ denotes the weight of the heaviest job scheduled during slot $s_j$. This follows because the algorithm of Fung et al. is $3$-competitive with an unbiased bit. 


      
The following observations hold.
\begin{enumerate}
          \item $\proc {Alg}_{<C_i}(\sigma) = \proc{Opt}_{<C_i}(\sigma)$;
          \item $\proc{Opt}(\sigma) = \proc{Opt}_{<C_i}(\sigma) + \proc{Opt}_{\geq C_i}(\sigma)$;
          \item $\proc{Opt}_{\geq C_i}(\sigma) \leq \frac{3}{2}(\sum_{\text{odd } j}W_j+\sum_{\text{even }j}W_j)$
      \end{enumerate}

      As before, we apply these observations to obtain the desired competitive ratio.

      \begin{align*}
        \mathbb{E}[\proc{Alg}(\sigma)] &\geq (\sqrt2-1)(\proc{Alg}_{<C_i}(\sigma)+\sum_{\text{odd $j$}}W_j)+(\sqrt2-1)(\proc{Alg}_{<C_i}(\sigma)+\sum_{\text{even $j$}}W_j) \\
        &= (\sqrt2-1)(2\proc{Alg}_{<C_i}(\sigma)+\sum_{\text{odd $j$}}W_j+\sum_{\text{even $j$}}W_j) \\
        &\geq (\sqrt2-1)(2\proc{Opt}_{<C_i}(\sigma)+\frac{2}{3}\proc{Opt}_{\geq C_i}(\sigma)) \\
        &\geq \frac{2(\sqrt2-1)}{3}\proc{Opt}(\sigma)
    \end{align*}

    The second inequality follows from Observations (1) and (3), and the third from Observation (2). Hence, the derandomization is $\frac{3}{2\sqrt2-2}\approx3.621$-competitive.
\end{proof}

\subsection{Derandomizing the Equal-Length Throughput Algorithm of Chrobak et al. \cite{chrobak2007online}}
\label{appendix:chrobak}

Chrobak et al. present a $\frac{5}{3}$-competitive 1-bit barely random algorithm in the model that allows no recourse once jobs begin. The competitive ratio of the derandomized algorithm depends only on the quality of the extracted bit. If the bit is unbiased, then our algorithm exactly matches the $\frac{5}{3}$-competitiveness result. Otherwise, we show that the derandomization is $\frac{5}{2\sqrt2}\leq 1.768$-competitive, and that the ratio monotonically decreases as the bias approaches $\frac{1}{2}$. The resulting algorithm is deterministic in the ROM and beats the deterministic lower bound of 2 in the adversarial model.
\subsubsection{Definitions and Notation} We begin by restating some definitions from \cite{chrobak2007online}, adapting them slightly for our analysis. We express each job $J_i=(r_i,p_i,s_i,w_i)$ in terms of its slack and note that we consider the setting where $w_i=1$ for all jobs.

The \textit{expiration time} of job $J_i$ is $x_i=r_i+s_i=d_i-p_i$. At time $t$, job $J_i$ is: \textit{admissible} if $r_i \leq t \leq x_i$, and \textit{pending} if it is both admissible and incomplete. Note that an executing job may also be pending.


A set of jobs $Q$ is \textit{feasible} at time $t$ if every job $J_i \in Q$ can still complete by its deadline $d_i$, and is \textit{flexible} if it is feasible at time $t+p$. If $Q$ is not flexible, then it is \textit{urgent}. A job started at time $t$ is \textit{flexible} if the set of pending jobs at time $t$ is flexible, and \textit{urgent} otherwise. We fix an ordering $\prec$ on any set of jobs such that $J_i\prec J_k$ implies $d_i \leq d_k$, breaking ties arbitrarily but consistently. The \textit{earliest-deadline} (ED) job of a set is a $\prec$-minimal job in the set.

For a schedule $A$, let $S_i^A$ and $C_i^A=S_i^A+p_i$ denote the \textit{start} and \textit{completion times} of job $J_i$. A completed schedule $A$ is \textit{earliest-deadline-first} (EDF) if whenever it starts a job, that job is the ED job among all pending jobs that are eventually completed in $A$. Note that this property can only be verified in hindsight since it depends on the final schedule of $A$. (Thus, $A$ is not obligated to choose the overall ED pending job).

A schedule $A$ is \textit{normal} if: whenever it starts a job, it chooses the ED job from the set of all pending jobs, and is executing some job at all times when the set of pending jobs is urgent.

Any normal schedule is EDF, but not every EDF schedule is normal. An EDF schedule may schedule the ED job among jobs that it will complete instead of the ED job among all pending jobs.  Two schedules $A$ and $A'$ are \textit{equivalent}, if for each time $t$: $A$ starts a job $J_i$ at $t$ if and only if $A'$ starts a job $J_k$ at $t$, and the job chosen by $A$ is flexible if and only if the job chosen by $A'$ is flexible. If both jobs are flexible, then $J_i=J_k$; otherwise, both are urgent and $J_i$ and $J_k$ may differ.

The randomized algorithm of Chrobak et al. executes two processes synchronized via a lock. Both processes follow the same instructions (Algorithm \ref{alg:process}). 

\begin{algorithm}\footnotesize
\caption{The Process Algorithm.}
\label{alg:process}
\begin{algorithmic}
\Procedure{Process}{$P,t$}
\State{$Q_t^P \gets$ set of pending jobs at time $t$, given schedule $P$.}
\If {$Q_t^P \neq \emptyset$}
    \If {$Q_t^P$ is urgent} \State {Execute the ED job in $Q_t^P$.}
    \ElsIf {$Q_t^P$ is flexible and the lock is available}
        \State {Acquire the lock.}
        \State {Execute the ED job in $Q_t^P$.}
        \State {Release the lock upon completion.}
    \Else \State {Idle.}
    \EndIf
\EndIf
\EndProcedure
\end{algorithmic}
\end{algorithm}

If both processes simultaneously request the lock, the algorithm of Chrobak et al. breaks ties arbitrarily. In our derandomized algorithm, we assume that the lock is consistently allocated to one of the two processes (denoted $X$ in the analysis), which we say is given \textit{lock priority}. The randomized algorithm of Chrobak et al. executes two processes in parallel from time $0$. Before either process executes, it flips a coin to decide which process produces the schedule. With a fair coin, this algorithm is $\frac{5}{3}$-competitive. We present a derandomized algorithm in the ROM that simulates their algorithm by replacing the fair coin flip with a bit extracted from the random-order input via \texttt{COMBINE}.

\subsubsection{The Derandomized Algorithm in the ROM} Our algorithm constructs the schedule in two phases. 
\begin{enumerate}
    \item \textit{Phase 1.} The schedule begins by greedily packing jobs pseudo-identical to $J_1$ until the first non-pseudo-identical job (having slack different from $s_1$) arrives. This phase is executed by a single process, unlike the synchronized processes of Chrobak et al., and leaves no idle time unless there are no pending jobs.
    \item \textit{Phase 2.} The remaining jobs are scheduled using the algorithm of Chrobak et al., with the coin flip simulated using the bit extracted via \texttt{COMBINE}. We switch to this phase once the first distinct job arrives.
\end{enumerate}

\begin{algorithm}\footnotesize
\caption{A helper procedure used in \proc{Rom-Simulation}.}\label{alg:begin-synch-helper}
\begin{algorithmic}
\Procedure{Begin-Synchronized}{$b, t, S$}
\State Begin the synchronized processes at time $t$, given $S$. It is possible that $t$ is a time strictly
\State less than $p$ units in the past. If this is the case, then no jobs release between $t$ and just
\State before the current time (see Algorithm \ref{alg:romalg}) so use $Q_{t}^S$ to initialize the state of the two 
\State processes. 
\If{$b=1$}
    \State Execute the process that is given lock priority in tie-breaking and simulate the other.
\Else
    \State Execute the process that is not given lock priority in tie-breaking and simulate the
    \State other.
\EndIf
\EndProcedure
\end{algorithmic}
\end{algorithm}

\begin{algorithm}\footnotesize
\caption{The ROM Algorithm.}
\label{alg:romalg}
\begin{algorithmic}
\Procedure{Rom-Simulation}{}
\State{Let $S$ denote the schedule that this algorithm will iteratively construct.}
        \If{only jobs pseudo-identical to $J_1$ have released so far}
        \State{$Q_t^S\gets$ set of pending jobs at time $t$, given the schedule $S$.}
        \If{$S$ is idle at time $t$ and $Q_t^{S}\neq \emptyset$}
        \State {Execute the ED job in $Q_t^S$.}
        \EndIf
\Else
\If {the first job not pseudo-identical to $J_1$ releases at time $r$}
\State {Extract a bit $b$ with worst-case bias $2-\sqrt2$ using \texttt{COMBINE}.}
\If {$S$ is idle at time $r$}
\State {$B \gets r$}
\State \Call{Begin-Synchronized}{$b,B,S$}
\Else 
\State {$B \gets $ start time of the currently executing job.}
\State {$Q_{B}^S \gets $ pending jobs at time $B$, given $S$.}
\If {$b=1$, \textbf{or} $b=0$ and $Q_{B}^S$ is urgent}
\State \Call{Begin-Synchronized}{$b,B,S$}
\ElsIf {$b=0$ and $Q_{B}^S$ is flexible} 
\State {Let $J_k$ denote the currently executing job. We introduce a special copy}
\State{$J_{k'}$ of $J_k$ with identical release time, deadline, and processing time, with}
\State{the property that when scheduled, it forces the machine to idle for the}
\State{duration of its execution.}
\Statex
\If{$Q_t^S$ is flexible at all times $t \in [B,B+p)$}
\State {Add $J_{k'}$ to $Q_{B+p}^S$ if it is still pending at time $B+p$.}
\State \Call{Begin-Synchronized}{$b,B+p,S$}
\ElsIf{$Q_t^S$ first becomes urgent at some time $t \in (B,B+p)$}
\State {Add $J_{k'}$ to $Q_{t+p}^S$ if it is still pending at time $t+p$.}
\State {Let $J_\ell$ denote the ED job in $Q_t^S$.}
\State {Remove $J_\ell$ from $Q_{t+p}^S$ if it is still pending at time $t+p$.}
\Statex
\State {Idle on $[B+p,t+p)$.}
\State \Call{Begin-Synchronized}{$b,t+p,S$}
\EndIf
\EndIf
\EndIf
\EndIf
\EndIf
\EndProcedure
\end{algorithmic}
\end{algorithm}

In the derandomized algorithm (Algorithm \ref{alg:romalg}), we let $S$ denote the schedule produced by the real machine. Let $X$ (respectively, $Y$) denote the schedule produced by Algorithm \ref{alg:romalg} when $b=1$ (respectively, $b=0$.) When Phase 2 begins, $S$ (and thus $X$ and $Y$) may be executing a job. Depending on the state of the pending queue at the time when the transition occurs, it may be that process without lock priority would have preferred to idle instead. A primary technical challenge is accounting for this without interrupting the running job, and the remaining task is to argue that the schedule produced by the derandomized algorithm has sufficient structural properties to inherit the analysis of \cite{chrobak2007online}.

Our algorithm illustrates a general method for derandomizing 1-bit barely random algorithms. In the first phase, greedily process pseudo-identical inputs until the first distinct input arrives. Then extract a bit and begin a fresh execution of the original 1-bit algorithm in the second phase. The analysis of the 1-bit algorithm applies to the remaining input. Ideally, the greedy phase extends an optimal solution, allowing us to recover the original competitive ratio when the extracted bit is fair. 

\subsubsection{Analysis} We begin by observing that the greedy packing in Phase 1 is optimal when all jobs are pseudo-identical (in which case no bit can be extracted by \texttt{COMBINE}.)

\begin{lemma}\label{lem:packingopt}
Phase 1 is optimal for sequences containing only pseudo-identical jobs. 
\end{lemma}

\begin{proof}
    The proof is by an elementary shifting argument. Phase 1 leaves no idle time unless no pending jobs remain, and whenever the derandomized algorithm schedules a job, it selects the job with the earliest deadline among all pending jobs.
\end{proof}

The breakpoint $B$ marks the transition from Phase 1 to Phase 2. Next, we show that some optimal solution extends the partial solution produced by Phase 1. Let $G$ denote the prefix of $S$ on times $[0,B)$ (with $B \gets \infty$ if all jobs are pseudo-identical.) 

\begin{lemma}\label{lem:optext}
  There exists an optimal schedule that extends $G$ (i.e., is equivalent to $G$ on the interval $[0,B)$).
\end{lemma}

\begin{proof}
If every job is pseudo-identical, then $G$ is optimal by \cref{lem:packingopt}. 

Otherwise, 
let $\proc{Opt}$ be an optimal offline schedule and assume without loss of generality that $\proc{Opt}$ schedules jobs pseudo-identical to $J_1$ by earliest-deadline-first. We transform $\proc{Opt}$ into another optimal schedule as follows. For each job started by $\proc{Opt}$ before $B$, shift it so that the $i^{th}$ such job begins at the same time as the $i^{th}$ job started by $G$. All such jobs are pseudo-identical to $J_1$ since they are released before time $B$. 

If after shifting all such jobs, some $j^{th}$ job beginning before $B$ in $G$ is unmatched, then $\proc{Opt}$ must schedule a job pseudo-identical to $J_1$ at or after $B$. Let $J$ be such a job with earliest deadline and schedule it instead when the $j^{th}$ job in $G$ begins. Repeating this process yields an optimal schedule equivalent to $G$ on $[0,B)$. 

Note that $G$ schedules at least as many jobs as $\proc{Opt}$ before time $B$; otherwise, after matching every job in $G$, $\proc{Opt}$ would start another job before $B$ that $G$ does not, contradicting the eager nature of $G$. 
\end{proof}


For analysis, define a subinstance $J'$ as follows. For each job pseudo-identical to $J_1$ that is pending at time $B$, set its release time to $B$, and include in $J'$ all jobs that release at or after $B$.


\begin{lemma}\label{lem:algext} If $X'$ and $Y'$ denote the schedules produced by the algorithm of Chrobak et al. on $J'$, then the schedules produced by Algorithm \ref{alg:romalg} satisfy $X=G \cup X'$ and $|Y| \geq |G \cup Y'|-1$.
\end{lemma}

\begin{proof}
$X$ and $Y$ match $G$ on $[0,B)$. Hence, we consider their behaviour on times $[B,\infty)$.

If the first job not pseudo-identical to $J_1$ arrives while $S$ is idle, then $Q_{B}^S$ is empty because otherwise Phase 1 would have started a job. Hence, at time $B=r$ where $r$ is the release time of this distinct job, both processes begin with identical pending queues of jobs in $J'$ and the derandomized algorithm will produce $X=G \cup X'$ or $Y=G \cup Y'$.

On the other hand, if the first job not pseudo-identical to $J_1$ arrives at time $r$ as $S$ is executing some other job $J_k$, then $B=S^S_k$ and we consider two cases.

\textbf{Case 1.} If $b=1$, or $b=0$ and $Q_{B}^S$ is urgent, then both processes start $J_k$ at time $B$. The pending queues of both processes are identical at $B$. Hence, executing the selected process from time $B$ onward will yield $X=G \cup X'$ or $Y=G \cup Y'$.

\textbf{Case 2.} If $b=0$ and $Q_B^S$ is flexible, then the selected process would idle at time $B$, whereas the derandomized algorithm has irrevocably started a job $J_k$ at $B$. We consider two subcases, and in each case, make use of a \textit{dummy} job $J_{k'}$ that inherits the release time, deadline, and processing time of $J_k$, but has the property that it will force the machine to idle for the duration of its execution. Let the selected process, had it begun from time $B$, be known as the \textit{model process}.

\textbf{Subcase 2.1.} If there is no time $t \in (B,B+p)$ at which $Q_t^S$ becomes urgent, then the model process idles on $[B,B+p)$ and may schedule $J_k$ in the future. The copy $J_{k'}$ that we add to $Q_{B+p}^S$ ensures that $Q_{B+p}^S$ matches the pending queue of the model process at time $B+p$. Hence, the pending queues match, except that $J_{k'} \in Q_{B+p}^S$ if and only if $J_k$ is pending in the model process at time $B+p$. Therefore, the derandomized algorithm and the model process behave identically after time $B+p$, except possibly when the derandomized algorithm idles while executing $J_{k'}$, whereas the model process executes $J_k$. Hence, $|Y|\geq |G \cup Y'|$ with equality when $Y'$ schedules $J_k$, in which case $Y=G \cup Y'$ by shifting $J_k$ to run in the idle space created by $J_{k'}$.

\textbf{Subcase 2.2.} If there is some time in $(B,B+p)$ where the pending queue of $S$ becomes urgent, let $t \in (B,B+p)$ be the first such time. Let $J_\ell$ denote the job that the model process schedules at time $t$. The derandomized algorithm continues executing $J_k$, which conflicts with $J_\ell$. Hence, once $J_\ell$ completes at time $t+p$, the pending queue of the model process contains $J_k$ if it has not expired, but not $J_\ell$. On the other hand, $Q_{t+p}^S$ contains $J_{k'}$ if it has not expired and removes $J_\ell$ if it is still pending. Hence, $Q_{t+p}^S$ matches the pending queue of the model process at time $t+p$, and both the derandomized algorithm and the model process behave identically after time $t+p$, except possibly when the model process executes $J_k$ while the derandomized algorithm executes $J_{k'}$ and idles. If $Y'$ never schedules $J_k$, then $|Y|=|G \cup Y'|$ since $Y$ replaces $J_\ell$ in $Y'$ with $J_k$. Otherwise, if $Y'$ schedules $J_k$, then the derandomized algorithm executes $J_{k'}$, causing the machine to idle during $[S^{Y'}_k,C_k^{Y'})$ and possibly losing a job. Hence, $|Y| \geq |G \cup Y'|-1$.  

\end{proof}

By Lemma \ref{lem:optext}, let $\proc{Opt}=G \cup \proc{Opt}'$ denote an optimal schedule that extends $G$. Here, $\proc{Opt}'$ is the optimal schedule from time $B$ onward. Since $X,Y$ and $\proc{Opt}$ share the prefix $G$ prior to $B$, all three schedules have exactly the same pending jobs at $B$. Hence, $X', Y'$ and $\proc{Opt}'$ are schedules over $J'$. 

Note that $X'$ and $Y'$ are constructed entirely by Algorithm \ref{alg:process} and hence are normal by construction. We state a lemma and a theorem from \cite{chrobak2007online} without proof, as both are crucial for the following analysis.

\begin{lemma}[Lemma 2.2 \cite{chrobak2007online}]\label{lem:EquivEDF} Let $X$ be a normal schedule for a set of jobs $J$. Let $f: J \to J$ be a partial function such that whenever $f(J_i)$ is defined, $J_i$ is scheduled as flexible in $X$ and $r_{f(J_i)} \leq C_i^X \leq x_{f(J_i)}$. 

There exists an EDF schedule $A$ equivalent to $X$ such that the following hold:
\begin{enumerate}
    \item All jobs $J_{f(J_i)}$ are completed in $A$;
    \item Consider a time $t$ when either $A$ is idle or it starts a job and the set of its pending jobs is feasible at $t$. Then all jobs pending at $t$ are completed in $A$. 
\end{enumerate}
\end{lemma}

\begin{theorem}[Theorem 3.1 \cite{chrobak2007online}]\label{thm:charge}
    Under the charging scheme of \cite{chrobak2007online}, each job in $X'$ or $Y'$ receives at most $\frac{5}{6}$ units of charge from jobs in $\proc{Opt}'$.   
\end{theorem}

The proof of Theorem \ref{thm:charge} relies on the fact that $X'$ and $Y'$ are normal schedules. We refer the reader to \cite{chrobak2007online} for the full charging argument.

Let $Y_B$ (respectively, $X_B$) denote the suffix of $Y$ (respectively, $X$) starting at time $B$. Note that $X_B=X'$ by Lemma \ref{lem:algext}. An immediate corollary is that each job in $X_B$ receives at most $\frac{5}{6}$ units of charge from jobs in $\proc{Opt}'$, and each job in $Y_B$ receives at most $\frac{5}{6}$ units of charge from jobs in $\proc{Opt}'$, with the possible exception of a single job in $Y_B$. 

\begin{corollary}\label{cor:charge}
    Under the charging scheme of \cite{chrobak2007online}, each job in $X_B$ receives at most $\frac{5}{6}$ units of charge from jobs in $\proc{Opt}'$. The same holds for $Y_B$, except a single job may receive at most $\frac{10}{6}$ units of charge from jobs in $\proc{Opt}'$.
\end{corollary}
\begin{proof}
    Each job in $X_B=X'$ receives at most $\frac{5}{6}$ units of charge from jobs in $\proc{Opt}'$ by Theorem \ref{thm:charge}. Next, consider the charging scheme for $Y'$. Let $J_k$ and $J_\ell$ be the same jobs as in Algorithm \ref{alg:romalg}. If $Y'$ schedules $J_\ell$, then redirect all charge assigned to $J_\ell$ in $Y'$ to $J_k$ in $Y_B$. The remaining jobs in $Y_B$ receive exactly the same charge as the corresponding jobs in $Y'$. If $Y'$ also schedules $J_k$, then both the charge assigned to $J_\ell$ and the charge assigned to $J_k$ in $Y'$ are redirected to $J_k$ in $Y_B$. Since each job in $Y'$ receives at most $\frac{5}{6}$ units of charge, the job $J_k$ in $Y_B$ may receive at most $\frac{10}{6}$ units of charge from jobs in $\proc{Opt}'$.    
\end{proof}

\subsubsection{Competitive Ratio}\label{appendix:chrobak-competitive-ratio}
\begin{lemma}
    $G \cup X'$ and $G \cup Y'$ are both normal.
\end{lemma}

\begin{proof}
By construction, every job that begins during $G$ has the earliest deadline among pending pseudo-identical jobs. Further, $G$ leaves no idle time when there exist pending jobs to be executed, so it remains busy whenever the pending queue is urgent. The rest of the schedule is constructed entirely by Algorithm \ref{alg:process} and is therefore normal by construction.
\end{proof}

\begin{lemma}\label{lem:factor2}
    For any two normal schedules $A$ and $B$, $\frac{1}{2} \leq \frac{|A|}{|B|} \leq 2$.
\end{lemma}
\begin{proof}
We compare $A$ and $B$ to an optimal schedule $\proc{Opt}$ on the same instance and show, via a charging argument, that each job in either schedule receives at most two units of charge from jobs in $\proc{Opt}$. Each job in $\proc{Opt}$ generates one unit of charge. We combine this with the observation that $|A|, |B| \leq |\proc{Opt}|$ to complete the bound. 

First, construct an equivalent EDF schedule $A'$ using \cref{lem:EquivEDF}. Consider any job $J_{i}$ in $\proc{Opt}$. If $A'$ is running a job $J_k$ at time $S_{i}^{\proc{Opt}}$, then assign one unit of charge from $J_i$ to $J_k$. If $A'$ is idle at time $S_{i}^{\proc{Opt}}$, then $J_{i}$ must be completed or pending in $A'$. If it is pending, then by \cref{lem:EquivEDF}.2, $J_{i}$ will be completed in $A'$. In this case, charge $J_i$ to itself. Each job in $A'$ receives at most 2 units of charge, so by the equivalence of $A$ and $A'$, it follows that $|\proc{Opt}| \geq |A| \geq \frac{1}{2}|\proc{Opt}|$. The same bound for $|B|$ follows by symmetry.
\end{proof}

\begin{theorem}
    The derandomized algorithm, Algorithm \ref{alg:romalg}, is $\frac{5}{2\sqrt2}\leq 1.768$-competitive in the ROM.
\end{theorem}
\begin{proof}

Note that in all cases, $X=G \cup X'$. For ease of exposition, we adopt the terminology introduced in Algorithm \ref{alg:romalg} and Lemma \ref{lem:algext}.

\textbf{Case 1.} If the first job not pseudo-identical to $J_1$ arrives while $S$ is idle or $Q_B^S$ is urgent, then $Y=G \cup Y'$. By Theorem \ref{thm:charge},

\begin{align*}
|\proc{Opt}|=|G|+|\proc{Opt}'| &\leq |G|+\frac{5}{6}(|X'|+|Y'|) \\
&\leq \frac{10}{6}|G|+\frac{5}{6}|X_B| + \frac{5}{6}|Y_B| \\
&= \frac{5}{6}|X| + \frac{5}{6}|Y| \\
&= \frac{5}{6}|G \cup X'|+ \frac{5}{6}|G \cup Y'| = \beta((2-\sqrt2)|G \cup X'|+(\sqrt2-1)|G \cup Y'|)
\end{align*}
for some $\beta>1$. To solve for $\beta$, rewrite the final equality as $\beta=\frac{\frac{5}{6}|G\cup X'|+\frac{5}{6}|G \cup Y'|}{(2-\sqrt2)|G\cup X'|+(\sqrt2-1)|G \cup Y'|}$. Next, set $r=\frac{|G\cup X'|}{|G\cup Y'|}$ and rewrite the bounded ratio as
\begin{align*}
    \beta=\frac{\frac{5}{6}(r+1)}{(2-\sqrt2)r+\sqrt2-1} \leq \frac{5}{2\sqrt2}
\end{align*}

which we maximize over $r \in [\frac{1}{2},2]$ by Lemma \ref{lem:factor2} since both $G\cup X'$ and $G\cup Y'$ are normal. Since $Y=G\cup Y'$ in this case, $\frac{|\proc{Opt}|}{(2-\sqrt2)|G\cup X'|+(\sqrt2-1)|G \cup Y'|} \leq \beta \leq \frac{5}{2\sqrt2} \approx 1.768$.

\textbf{Case 2.} If $Q_{B}^S$ is flexible, then we consider two subcases.

\textbf{Subcase 2.1.} If there is no time $t \in (B,B+p)$ at which $Q_t^S$ becomes urgent, then $|Y| = |G\cup Y'|$ if $Y'$ schedules $J_k$ and $|Y|=|G\cup Y'|+1$ otherwise.

If $Y'$ schedules $J_k$, then the analysis for $|Y|=|G\cup Y'|$ exactly matches Case 1. Otherwise, by Theorem \ref{thm:charge},

\begin{align*}
    |\proc{Opt}| \leq \frac{5}{6}|G \cup X'|+ \frac{5}{6}|G \cup Y'| &= \beta((2-\sqrt2)|G \cup X'|+(\sqrt2-1)|G \cup Y'|) \\
    &= \beta((2-\sqrt2)|X|+(\sqrt2-1)(|Y|-1))
\end{align*}
for some $\beta>1$. Again, we solve for $\beta$ via the first equality. Then,

\begin{align*}
    \frac{|\proc{Opt}|}{(2-\sqrt2)|X|+(\sqrt2-1)|Y|} \leq \frac{|\proc{Opt}|}{(2-\sqrt2)|X|+(\sqrt2-1)(|Y|-1)} \leq \beta
\end{align*}

so the derandomized algorithm is once again $\beta \leq \frac{5}{2\sqrt2}\approx 1.768$-competitive.

\textbf{Subcase 2.2.} If there is some time $t \in (B,B+p)$ where the pending queue of $S$ becomes urgent, then $|Y| = |G \cup Y'|$ if $Y'$ does not schedule $J_k$ and $|Y| = |G\cup Y'|-1$ otherwise. 

If $Y'$ does not schedule $J_k$, then the analysis for $|Y|=|G\cup Y'|$ exactly matches Case 1. Otherwise, by Corollary \ref{cor:charge},

\begin{align*}
        |\proc{Opt}| \leq \frac{5}{6}|X|+ \frac{5}{6}(|Y|+1) &= \frac{5}{6}|G\cup X'|+\frac{5}{6}|G\cup Y'| \\
        &=\beta((2-\sqrt2)|G \cup X'|+(\sqrt2-1)|G \cup Y'|) \\
    &= \beta((2-\sqrt2)|X|+(\sqrt2-1)(|Y|+1))
\end{align*}

As before, solve for $\beta$ via the middle equality. Then,
\begin{align*}
    |\proc{Opt}| \leq \beta((2-\sqrt2)|X|+(\sqrt2-1)(|Y|+1))=\beta((2-\sqrt2)|X|+(\sqrt2-1)|Y|)+\beta(\sqrt2-1)
\end{align*}
where $\beta(\sqrt2-1)=o(|\proc{Opt}|)$.  Hence, the derandomized algorithm is (asymptotically) $\beta \leq \frac{5}{2\sqrt2}\approx 1.768$-competitive.



\end{proof}

\subsection{Derandomizing the Throughput Algorithm of Kalyanasundaram and Pruhs \cite{kalyanasundaram2003maximizing}}
\label{appendix:KP}


In the real-time model with preemption (i.e., jobs may be interrupted and resumed later), Kalyanasundaram and Pruhs \cite{kalyanasundaram2003maximizing} present a barely random algorithm with constant competitiveness. While no deterministic online algorithm can achieve a constant competitive ratio against an adversarial input sequence \cite{baruah1994line}, they show that the optimal offline solution is bounded above by a linear combination of the solutions produced by two deterministic algorithms. 

We begin by presenting the two algorithms that complete the bound. The first algorithm, \proc{Srpt}, always executes the job with the least remaining processing time and preempts jobs as necessary whenever new jobs are released. When multiple jobs have the same remaining processing time, \proc{Srpt} breaks ties arbitrarily. In our derandomization, we assume that \proc{Srpt} breaks ties by selecting the job with the earliest deadline. The second algorithm, \proc{Lax}, is slightly more sophisticated and maintains two stacks: one consisting of jobs that are currently scheduled to run on the machine, and another consisting of jobs that it may consider scheduling. Jobs in the latter stack must, in addition to feasibility, satisfy eligibility constraints depending on their slack. We refer to the former stack as the \textit{promised} stack and the latter as the \textit{pending} stack.

We note that as long as the jobs arriving online are pseudo-identical (i.e., have the same processing time and slack), scheduling feasible jobs in order of nondecreasing release times maximizes the number of these jobs completed by their deadlines by a simple exchange argument. \proc{Srpt} breaks ties by earliest deadline (which is equivalent to ordering by release times for pseudo-identical jobs), and is therefore optimal for scheduling pseudo-identical jobs. On the other hand, the eligibility criteria of the pending stack in \proc{Lax} may lead it to skip some jobs depending on the ratio between the slack and processing time of the pseudo-identical jobs. For example, consider two jobs $J_1$ and $J_2$ with $r_1=0$, $r_2=0.6$, $p_1=p_2=1$, and $s_1=s_2=0.4$. $J_2$ is available at time $1$ when $J_1$ completes and is both started and completed by $\proc{Srpt}$. In contrast, $J_2$ is no longer viable (see Definition 1 in \cite{kalyanasundaram2003maximizing}) for $\proc{Lax}$ at time $1$ and is left unscheduled by $\proc{Lax}$.

We derandomize the 1-bit algorithm that chooses uniformly between \proc{Srpt} and \proc{Lax} by executing \proc{Srpt} until a distinct job arrives. The derandomization guarantees that if the algorithm switches to \proc{Lax}, the final schedule completes at least as many jobs as \proc{Lax} (run from time $0$) would have, and never attempts to schedule a job that \proc{Srpt} has already completed.

\begin{algorithm}[H]\footnotesize
\caption{A ROM Algorithm for the unweighted throughput problem (with preemption.)}
\label{alg:KP-throughput}
\begin{algorithmic}
\State Execute \proc{Srpt} until a distinct job arrives at time $r$. As this happens, simulate the decisions that \proc{Lax} would have made if executed from time $0$. Note that the state of its stacks is fully determined by its decisions on the input. 
\State Extract a bit $b$ with worst-case bias $2-\sqrt2$.
\State
\If {$b=1$}
\State Continue scheduling via \proc{Srpt}. 
\Else
\If {\proc{Lax} is idle at time $r$}
\State Preempt the currently executing job, if it exists.
\State Switch to scheduling via \proc{Lax}.
\Else
\If {\proc{Srpt} is idle at time $r$}
\State Wait until the currently running job in the simulated \proc{Lax} schedule completes or is preempted
\State Switch to \proc{Lax}, keeping the current contents of its stacks. Whenever \proc{Lax} attempts to schedule
\State a copy of a pseudo-identical job released before $r$, idle until that job is preempted or completed. 
\State Follow all other scheduling decisions of \proc{Lax}. 
\Else 
\State Switch to \proc{Lax}, keeping the current contents of its stacks. Whenever \proc{Lax} attempts to schedule 
\State a copy of a pseudo-identical job released before $r$, execute the earliest-deadline copy of that \State job released before $r$ (i.e., the one that released earliest) that is still incomplete. If no such \State job remains, then idle until that job in \proc{Lax} is preempted or completed.
\State Follow all other scheduling decisions of \proc{Lax}. 
\EndIf
\EndIf
\EndIf
\end{algorithmic}
\end{algorithm}

\begin{theorem}
     Derandomizing the unweighted throughput (with preemption) algorithm of Kalyanasundaram and Pruhs \cite{kalyanasundaram2003maximizing} yields a deterministic $\frac{129024}{\sqrt2-1}\approx 311491.491$-competitive algorithm for the unweighted throughput problem (with preemption) in the ROM.
\end{theorem}
\begin{proof}
    When $b=1$, the algorithm executes \proc{Srpt} from time $0$. If $b=0$, then we consider the following cases.

    \textbf{Case 1.} If \proc{Lax} is idle when the distinct job arrives, the algorithm preempts the currently executing job, if it exists, and switches to \proc{Lax}. Since \proc{Lax} is idle, its stacks must be empty by induction: if the promised stack is nonempty then \proc{Lax} would be busy, and the pending stack is empty because one job would have been promoted to the promised stack on the completion of the last job by the stack update rules of \proc{Lax}. If there were no such jobs in the pending stack eligible for promotion, then time only reduces its eligibility, i.e., no jobs that were in the pending stack will re-enter it between the completion of a pseudo-identical job and the release of the next job.
    Hence, jobs that arrive from time $r$ onward are independent of the algorithm's actions before $r$. We can continue scheduling with \proc{Lax} without worrying about the past. Since \proc{Srpt} is optimal on pseudo-identical jobs, it follows that just before time $r$, \proc{Srpt} has completed at least as many jobs as \proc{Lax} (excluding the job that the algorithm might preempt.) This guarantees that the number of jobs that the derandomization completes is at least the number of jobs that \proc{Lax} would have completed had it been run from time $0$. 

    \textbf{Case 2.} If at time $r$ \proc{Lax} is executing while \proc{Srpt} is idle, then the greediness of \proc{Srpt} on pseudo-identical jobs implies that every pseudo-identical job that released before $r$ is either already completed by \proc{Srpt} or has deadline before $r$.
    Therefore, there is an injective correspondence from pseudo-identical jobs released prior to time $r$ that are completed by \proc{Lax}, to pseudo-identical jobs released prior to time $r$ that are completed by \proc{Srpt} before time $r$. In particular, whenever \proc{Lax} later attempts to schedule a pseudo-identical job released before time $r$ (that is still feasible), that job has already been accounted for by the derandomized algorithm. Our algorithm is therefore safe to ignore such requests from \proc{Lax}. For all other times, the scheduling decisions of \proc{Lax} and the derandomization will match. Hence, the number of jobs that our algorithm completes is at least the number of jobs that \proc{Lax} would have completed had it been run from time $0$. 

    \textbf{Case 3.} If at time $r$ both \proc{Lax} and \proc{Srpt} are executing, then they are executing pseudo-identical copies of the first job. By the optimality of \proc{Srpt} on pseudo-identical jobs, \proc{Srpt} has completed at least as many pseudo-identical jobs released before $r$ as \proc{Lax} by time $r$. Since all pseudo-identical jobs have equal processing time, it follows that there are no more total units of processing time that remain in pseudo-identical jobs released before $r$ under \proc{Srpt} than under \proc{Lax}. Also, since the algorithm schedules pseudo-identical jobs by nondecreasing release times, whenever \proc{Lax} schedules a pseudo-identical job released before time $r$, the derandomization must also be eligible to, unless no such jobs remain. The rest of the scheduling decisions (at all times other than when pseudo-identical jobs released before time $r$ are executing) in \proc{Lax} and the algorithm match. Since the algorithm stays ahead on pseudo-identical jobs, the number of jobs that our algorithm completes is at least the number of jobs that \proc{Lax} would have completed had it been run from time $0$. 

    Let \proc{OrigLax} and \proc{OrigSrpt} denote the number of jobs that \proc{Lax} and \proc{Srpt}, respectively, would have completed if they were run on the entire instance. Let \proc{OurLax} and \proc{OurSrpt} denote the number of jobs that our algorithm completes if $b=0$ and $b=1$, respectively. Finally, let $\mathbb{E}[\proc{Alg}]$ denote the expected number of jobs that our algorithm completes, and let \proc{Opt} denote the optimal offline number of jobs that can be completed given the realized ROM instance. 

    Kalyanasundaram and Pruhs \cite{kalyanasundaram2003maximizing} prove that $\frac{\proc{Opt}}{129024} \leq \proc{OrigSrpt}+\frac{108864}{129024} \proc{OrigLax} \leq \proc{OrigSrpt}+ \proc{OrigLax}$. Since the extracted bit has worst-case bias $2-\sqrt2$, the probability of choosing either algorithm is at least $\sqrt2-1$. Using this, we conclude that \[\mathbb{E}[\proc{Alg}] \geq (\sqrt2-1)(\proc{OurSrpt}+\proc{OurLax}) \geq (\sqrt2-1)(\proc{OrigSrpt}+ \proc{OrigLax}) \] so $\mathbb{E}[\proc{Alg}] \geq (\sqrt2-1)\frac{\proc{Opt}}{129024}$. Hence, the derandomization is $\frac{129024}{\sqrt2-1}\approx 311491.491$-competitive under random-order arrivals.
\end{proof}

With a fair bit we match their competitive ratio, while a biased bit incurs a constant multiplicative penalty. Importantly, the ratio is constant even with a biased bit and thus beats the deterministic lower bound under adversarial arrivals.

\section{Derandomizing Knapsack Algorithms in the Random-Order Model}

\subsection{Revocations are Necessary for $O(1)$-competitive Deterministic Knapsack Algorithms}
\label{appendix:revocations-required}

Any deterministic algorithm in the ROM must accept some subset of the initial identical items; otherwise, the adversary could present only identical items and keep the algorithm from packing anything.

The adversary can present a multiset consisting of many identical copies of a small item and a single copy of a large item so that the first item in the permuted sequence is small with high probability. This item must be packed by the ROM algorithm. By carefully choosing the sizes of the items, an adversary can force the ratio between the values of the two knapsacks to be arbitrarily large. The optimal solution accepts the large item while the ROM algorithm is stuck with small items.

\begin{theorem}\label{thm:det-ROM-prop-knapsack}
    Any deterministic ROM algorithm for the (general or proportional) knapsack problem that does not revoke has competitive ratio $\Omega(n)$.
\end{theorem}
\begin{proof}
Let $\proc{Alg}(\sigma)$ and $\proc{Opt}(\sigma)$ denote the values packed by a deterministic ROM algorithm and an optimal offline algorithm on a permuted input $\sigma$. 

Consider an input with $n-1$ small items of weight $\frac{2}{n(n-1)}$ and a single large item of unit weight, with the value of each item equal to its weight.
Suppose that $k\geq 0$ small items arrive before the large item. The probability of the first $k$ items being small is the probability that the large item appears at position $k+1$, which is $\frac{1}{n}$. When $k>0$, any deterministic algorithm must accept some nonempty subset of the first $k$ items (otherwise the adversary could present only identical items and force the algorithm to attain no value.) Therefore, $\mathbb{E}[\proc{Alg}(\sigma)] \leq \frac{1}{n}+\frac{1}{n}\sum_{k=1}^{n-1} k \cdot \frac{2}{n(n-1)}= \frac{2}{n}$, where the first term is the contribution from when the large item arrives first, and the summation accounts for the expected value packed when the first $k \geq 1$ items are small. On the other hand, $\proc{Opt}(\sigma)=1$. Therefore, the competitive ratio $\frac{\proc{Opt}(\sigma)}{\mathbb{E}[\proc{Alg}(\sigma)]} \geq \frac{1}{2/n}=\frac{n}{2}$ is $\Omega(n)$. 
\end{proof}



\subsection{Derandomizing the Non-Revoking Proportional Knapsack Algorithm of Han et al. \cite{han2015randomized} via Revoking}

We continue by describing the derandomization procedure referenced in Section \ref{sec:applications} for the non-revoking algorithm of Han et al. \cite{han2015randomized}.

\begin{algorithm}[H]\footnotesize
\caption{A ROM Algorithm for the Proportional Knapsack Problem.}
\label{alg:worse-ROM-prop-knapsack}
\begin{algorithmic}
\State Pack as many identical copies of the first item as possible. Let $w$ denote the weight of the first item.
\State When an item with a different weight arrives, extract a bit $b$ with worst-case bias $2-\sqrt2$.
\State Let $W$ denote the current weight of the knapsack. 
\State 
\If {$1-W<w$} 
\State Return the current knapsack. \EndIf
\State
\State Begin the algorithm of Han et al. with bin 1 equal to the current knapsack and bin 2 empty.
\If {$b=1$}
\State Return bin 1 from executing the algorithm of Han et al. on the rest of the input, i.e., pack as many
\State of the remaining items as can fit.
\Else
\State Revoke all items currently in the knapsack while continuing to simulate bin 1.
\State Return bin 2 from executing the algorithm of Han et al. on the rest of the input, i.e., if an inbound \State item fits in bin 1, then simulate accepting it in bin 1. If the item is too large to fit in bin 1 but fits \State into bin 2, accept it in bin 2 (the current knapsack.)
\EndIf
\end{algorithmic}
\end{algorithm}

Note that this algorithm packs exactly one of bin 1 or bin 2 from the algorithm of Han et al. Like the ROM derandomizations for interval selection and throughput scheduling, if the extracted bit is fair, we recover the original competitive ratio; otherwise if the bit is biased, the ratio suffers only due to the bias of the bit. 

\begin{theorem}
    Algorithm \ref{alg:worse-ROM-prop-knapsack} is a $\frac{1}{\sqrt2-1} \approx 2.414$-competitive deterministic algorithm (with revoking) for the proportional knapsack problem in the ROM.
\end{theorem}
\begin{proof}
If $1-W<w$, either $w \geq \frac{1}{2}$ or $w < \frac{1}{2}$. If $w \geq \frac{1}{2}$ then $W > 0$ implies that $W\geq w$. Otherwise, $w<\frac{1}{2}$ implies that $W>1-w>\frac{1}{2}$. In either case, the knapsack contains weight at least $\frac{1}{2}$, so the algorithm is $2$-competitive even without considering the rest of the input. 


On the other hand, if $1-W \geq w$, then bin 2 would be empty if we had run the algorithm of Han et al. from the beginning. If $b=1$, we can continue as-is and produce the same knapsack as bin 1 in Han et al.'s algorithm. Otherwise, we revoke all items currently in the knapsack and wait until Han et al.'s algorithm begins packing bin 2. 

Our algorithm chooses between bins 1 and 2 with probability at least $\sqrt2-1$ each. Since the total weight of the items in both bins is either optimal or exceeds the knapsack capacity, the derandomized algorithm is $\frac{1}{\sqrt{2}-1} \approx 2.414$-competitive.    
\end{proof}




\subsection{Derandomizing the Revoking Proportional Knapsack Algorithm of Han et al. \cite{han2015randomized}}
\label{appendix:han-prop-knapsack}

\subsubsection{Definitions and Notation}  We consider an input sequence $\sigma$ consisting of $n$ items and assume without loss of generality that the knapsack has unit capacity. The derandomized algorithm operates in the same model as \cite{han2015randomized}, which allows revocations -- that is, an item may be permanently discarded after being accepted into the knapsack. We begin by restating a few definitions from \cite{han2015randomized}.

The algorithm of Han et al. classifies the items into three groups: an item $i$ is \textit{small} if $w_i \leq \frac{3}{10}$, \textit{medium} if $\frac{3}{10}<w_i<\frac{7}{10}$, and \textit{large} if $w_i \geq \frac{7}{10}$. $S, M$ and $L$ denote the sets of small, medium and large items, respectively. We partition $M$ into four subsets $M_1$, $M_2$, $M_3$, and $M_4$ defined by $M_1 = \{i \in M : \frac{3}{10}<w_i\leq \frac{4}{10} \}$, $M_2 = \{i \in M : \frac{4}{10}<w_i \leq \frac{5}{10} \}$, $M_3 = \{i \in M : \frac{5}{10}<w_i<\frac{6}{10}\}$, and $M_4 = \{i \in M : \frac{6}{10} \leq w_i < \frac{7}{10}\}$. An item $i$ is said to be an \textit{$M_j$ item} if $i \in M_j$.  

$A_1$ and $A_2$ denote the subroutines defined in \cite{han2015randomized}, and $A$ refers to the 1-bit barely random algorithm that chooses between them uniformly at random. $A_1$ is aggressive, preferring to take heavier items first, and $A_2$ is balanced, accepting items from lighter classes first. In both, the arrival of a large item causes every other item to be evicted, and packing terminates. Let $Q$ denote the items currently in the knapsack together with the newly arrived item. We define \textit{round} $i$ (or the $i^{th}$ round) of an algorithm to denote the time between the arrival of the $i^{th}$ item to when the algorithm finishes processing it by readjusting its knapsack. We say that an algorithm \textit{keeps} an item if it is accepted or not discarded by its knapsack during the current round.

$\bm A_1$. On the arrival of a new item, if no $M_4$ item has been encountered so far, then $A_1$ keeps the smallest $M_3$ item in $Q$. Otherwise if an $M_4$ item has been seen, it keeps the smallest $M_4$ item in $Q$. If space allows, $A_1$ then keeps the largest items in $Q$ that keep the knapsack feasible.

$\bm A_2$. On the arrival of a new item, if no $M_4$ item has been seen and some feasible subset of $Q$ has total weight at least $\frac{9}{10}$, then $A_2$ stops packing and removes items from the knapsack so that it matches that subset. If some $M_4$ item has been observed and some feasible subset of $Q$ has total weight at least $\frac{8}{10}$, then $A_2$ stops packing and likewise adjusts the knapsack so that its contents match the chosen subset. Otherwise, $A_2$ keeps the smallest $M_2$ item and then the smallest $M_1$ item from $Q$. It then keeps the smallest medium items in $Q$, and if space allows, the heaviest small items in $Q$.

 The algorithm flips a coin to choose between $A_1$ and $A_2$. When the coin is fair, the resulting algorithm is $\frac{10}{7}$-competitive. A detailed description of the algorithm appears in \cite{han2015randomized}.
 
 The original algorithm tracks the knapsacks of both $A_1$ and $A_2$ and includes an early stopping condition: if the expected packed weight (i.e., half the sum of the two knapsack weights) exceeds $\frac{7}{10}$, the algorithm stops and packs nothing else. This is because the randomized algorithm may later reduce its packed weight and cease to be $\frac{10}{7}$-competitive. Consider a sequence of items with weights $(0.67,0.04,0.6)$. In round 2, both $A_1$ and $A_2$ keep items 1 and 2, and the randomized algorithm has expected packed weight $0.71$ (which is sufficient to be $\frac{10}{7}$-competitive.) If there were no stopping condition, then in round 3 both $A_1$ and $A_2$ would replace the $M_4$ item of weight $0.67$ with the $M_4$ item of weight $0.6$, and the algorithm would no longer be $\frac{10}{7}$-competitive. Since our derandomized algorithm makes a biased choice between $A_1$ and $A_2$, we must update the stopping threshold to reflect this bias, and stop if a weighted approximation of the expected packed weight exceeds the threshold in some round. In particular, we stop packing if $(2-\sqrt2)w_i(A_1)+(\sqrt2-1)w_i(A_2)>\frac{3\sqrt2-1}{5}$, where $w_i(A_j)$ is the weight of $A_j$'s knapsack immediately after round $i$ and $\frac{3\sqrt2-1}{5}$ is a constant that we justify later.

\subsubsection{A Derandomized Algorithm} We present a derandomized algorithm (Algorithm \ref{alg:rom-prop-knapsack}) that is $\frac{5}{3\sqrt2-1} \leq 1.542$-competitive in the ROM. Following our standard derandomization approach, we replace the fair bit in the original algorithm with a bit extracted from the stochastic input. We use the \texttt{COMBINE} procedure to extract a bit with bounded bias.

\begin{algorithm}\footnotesize
\caption{The Proportional Weights ROM Algorithm.}
\label{alg:rom-prop-knapsack}
\begin{algorithmic}
\State Accept as many copies of the first item as can fit in the knapsack. If the first item is large, then pack it and terminate the algorithm.
\State When the first item whose weight differs from $w_1$ arrives, extract a bit $b$ with worst-case bias $2-\sqrt2$.
\State
\If {$b=1$} \State Execute $A_1$ on the remainder of the input.
\Else \State Execute $A_2$ on the remainder of the input.
\EndIf
\end{algorithmic}
\end{algorithm}

A simple proof shows that Algorithm \ref{alg:rom-prop-knapsack} produces a knapsack identical to that produced by either $A_1$ or $A_2$ when all three algorithms are run on the same input. Hence, the only penalty to the competitive ratio is from the bias of the extracted bit.
\begin{lemma}
    Algorithm \ref{alg:rom-prop-knapsack} produces a knapsack identical to that produced by either $A_1$ or $A_2$ when all three algorithms are run on the same input sequence.
\end{lemma}
\begin{proof}
    We first show that all algorithms accept the same prefix of initial identical items. Then, once a distinct item arrives, the derandomized algorithm executes one of $A_1$ or $A_2$ on the remainder of the sequence and the claim follows.

    If the first item is large, then each algorithm accepts it and nothing else. Hence, the resulting knapsacks are identical.
    
    If the first item is medium, then $A_1$ takes as many such copies as possible and $A_2$ behaves similarly but may stop once a subset of the identical items has total weight at least $\frac{9}{10}$. If a distinct item arrives and the bit selects $A_2$, either $A_2$ has packed the same items as $A_1$ (and hence the derandomized algorithm up to this point), or it meets the weight threshold and terminates immediately. If instead the bit selects $A_1$, the derandomized algorithm executes $A_1$ on the remainder of the input and therefore matches $A_1$.
        
    On the other hand, if the first item is small, then $A_1$ and our derandomized algorithm accept as many identical copies of it as possible. $A_2$ does the same until its knapsack has total weight at least $\frac{9}{10}$. If the weight threshold was reached by packing identical inputs, then the derandomized algorithm matches $A_2$ because both the derandomized algorithm and $A_2$ have a subset of total weight at least $\frac{9}{10}$ when the distinct item arrives and therefore terminate immediately. Otherwise, the knapsacks of $A_1$, $A_2$, and the derandomized algorithm match up to the arrival of the first distinct item. If the weight threshold was not reached by the identical input prefix, then the derandomized algorithm executes the selected algorithm on the rest of the sequence and therefore matches it.
\end{proof}

The algorithm achieves its best performance when the knapsack produced by $A_1$ is selected with higher probability. This can be verified by repeating the subsequent analysis with the selection probabilities of $A_1$ and $A_2$ interchanged. Applying the cases of Lemmas 3, 4 and 5 from \cite{han2015randomized} then yields a competitive ratio of approximately $1.542$.

\begin{theorem}
    Algorithm \ref{alg:rom-prop-knapsack} is a $\frac{5}{3\sqrt2-1}$-competitive deterministic algorithm in the ROM.
\end{theorem}
\begin{proof}
    The proof follows by modifying the analyses of Lemmas 3, 4 and 5 from \cite{han2015randomized} to account for the biased selection probabilities.
    
    In Lemma 3, the ratios are maximized in cases (c) and (g.2). In these cases, the optimal knapsack has weight at most 1, while the derandomized algorithm's expected weight is given by a combination of the bounds derived in \cite{han2015randomized}. Hence, whenever one of the conditions in Lemma 3 holds, the ratio is bounded above by $\frac{1}{\frac{9}{10}(2-\sqrt2)+\frac{5}{10}(\sqrt2-1)} \leq 1.503$. The remaining cases are handled in Lemmas 4 and 5. 

    Lemma 4 considers the remaining cases under the assumption that the input contains no small items. Case 1 leads to a contradiction. Between Cases 2 and 3, the ratio is maximized in Case 3 (when the sequence contains an $M_4$ item) and is hence bounded above by $\frac{\frac{9}{10}}{\frac{13}{10}(\sqrt2-1)+\frac{6}{10}(3-2\sqrt2)} \leq 1.404$ (since in this case the optimal knapsack has weight at most $\frac{9}{10}$.)
    
    Similar to Lemma 4, Lemma 5 considers the cases not covered by Lemma 3 in which small items are present in the input. We first consider the cases in which both $A_1$ and $A_2$ accept all small items in the sequence, and the case in which the sequence contains no medium items. In the first case, Han et al. isolate the contribution of the small items and apply Lemma 4 to the rest of the sequence of (medium or large) items. By the above analysis, the competitive ratio is bounded above by $1.404$. When the sequence contains no medium items, both algorithms achieve a $\frac{10}{7}$-approximation, hence the derandomized algorithm is $\frac{10}{7}$-competitive independent of the bias of the extracted bit.

    The remaining non-contradictory cases are those in which the sequence contains medium items and exactly one of $A_1$ or $A_2$ rejects some small items in some $k$th round. In the first case, the knapsacks contain at least $\frac{4}{10}$ weight at the $k$th round, and one of the two knapsacks has weight at least $w_a+\frac{4}{10}$, where $w_a$ is the weight of the heaviest small item that the respective knapsack rejects in round $k$. The derandomized algorithm selects the knapsack of $A_1$ with probability at most $2-\sqrt2$ and that of $A_2$ with probability at least $\sqrt2-1$. Hence, the expected weight after round $k$ is either $(2-\sqrt2)(w_a+\frac{4}{10})+(\sqrt2-1)w_k(A_2)$ or $(\sqrt2-1)(w_a+\frac{4}{10})+(2-\sqrt2)w_k(A_1)$, where $w_k(A_i)$ is the weight of the knapsack of $A_i$ immediately after the $k$th round. In either case, the expectation is bounded below by $\frac{14}{10}(\sqrt2-1)+\frac{4}{10}(3-2\sqrt2)$, and hence the algorithm is $\frac{1}{\frac{14}{10}(\sqrt2-1)+\frac{4}{10}(3-2\sqrt2)}=\frac{5}{3\sqrt2-1} \leq 1.542$-competitive in this case. The remaining case (Case 2 in \cite{han2015randomized}) can be verified similarly and yields a smaller competitive ratio.
    
    Taking the maximum over all cases establishes that the derandomized algorithm is $\frac{5}{3\sqrt2-1} \leq 1.542$-competitive in the ROM. 
\end{proof}

This is strictly better than the tight deterministic bound of $\frac{1+\sqrt5}{2}\approx 1.618$ in the adversarial model.

\section{Derandomizing Makespan Algorithms in the Random-Order Model}

\subsection{Derandomizing the Two-Machine Job Shop Algorithm of Kimbrel and Saia \cite{KimbrelS00}}
\label{appendix:KimbrelSaia}

We begin by describing the 1-bit barely random algorithm of Kimbrel and Saia \cite{KimbrelS00}. The algorithm fixes an arbitrary ranking of the jobs and then selects between one of two deterministic algorithms with a fair coin flip. The first algorithm, \proc{Forward}, always executes the highest-ranked available job on machine 1 while machine 2 executes the lowest-ranked available job. We say that a job is \textit{available} on machine $i$ if its next (incomplete) operation requests machine $i$. The second algorithm, \proc{Reverse}, swaps the behaviors of the machines: machine 1 always executes the lowest-ranked available job and machine 2 always executes the highest-ranked available job. Preemptions are used as necessary to ensure that the currently executing job is either of highest or lowest rank (depending on the algorithm and machine.) Kimbrel and Saia show that in expectation, this 1-bit algorithm achieves a tight competitive ratio of 1.5 against oblivious adversaries. 

Next, we describe a derandomization of this algorithm in the ROM. First, fix an arbitrary ranking of the jobs and begin by executing \proc{Forward}. Operations are identical if their processing times are equal. When we observe two completed operations with unequal processing requirements, we extract a bit $b$ with a worst-case bias of $2-\sqrt2$. If $b=1$, then the algorithm continues executing \proc{Forward}. If $b=0$, the algorithm waits until the currently running operations complete and then switches to use \proc{Reverse}. If all operations have equal processing time, then our algorithm is 2-competitive from the observation that it keeps machines busy whenever jobs are available.

To simplify the analysis of the derandomization, we pad each job that has already completed some of its operations with dummy operations of the same processing requirements, ordered consistently with the original job structure. When the algorithm schedules a dummy operation, it forces the machine requested by the dummy operation to idle until it completes or is preempted. The addition of dummy operations is independent of the problem instance and in particular does not change the optimal offline makespan. 

This derandomization admits a loose bound but serves as proof that derandomization is possible. In practice, one would be better served by executing the na\"ive 2-competitive algorithm that keeps both machines busy.

\begin{theorem}
 Derandomizing the online two-machine job shop (with resuming) algorithm of Kimbrel and Saia \cite{KimbrelS00} yields a deterministic $10-5\sqrt2\leq 2.929$-competitive algorithm for the online two-machine job shop problem (with resuming) in the ROM.
\end{theorem}
\begin{proof}
    Let $\proc{Opt}$ denote the optimal offline makespan for the realized ROM sequence. 
    
    If $b=1$ then the makespan matches the solution, say $C_1$, of \proc{Forward} from \cite{KimbrelS00}. If $b=0$, the (prefix) time spent scheduling identical-length operations before switching to \proc{Reverse} is at most $2\proc{Opt}$: at least one of the two machines is active at all times unless every job has completed (in which case the algorithm terminates and we do not bother running $\proc{Reverse}$). The total processing times completed by machine 1 and machine 2 on the prefix are each a lower bound on $\proc{Opt}$ since any feasible schedule must process these operations on their respective machines. 
    
    Let $\mathbb{E}[\proc{Alg}]$ denote the expected makespan of our algorithm, and let $C_2$ denote the makespan obtained from running \proc{Reverse} on the ranking. Let $C_2'$ denote the makespan if $b=0$. By the above argument, $C_2' \leq C_2+2 \proc{Opt}$. Kimbrel and Saia \cite{KimbrelS00} prove that $C_1+C_2 \leq 3\proc{Opt}$. Thus, $C_1+C_2' \leq 5\proc{Opt}$, and \[\mathbb{E}[\proc{Alg}] \leq (2-\sqrt2)(C_1+C_2') \leq (2-\sqrt2)\cdot5\proc{Opt}\leq 2.929 \proc{Opt}\] 
\end{proof}

Attempting to leverage the stochastic nature of the input degrades competitive guarantees compared to simpler deterministic strategies; even against a fully-adaptive adversary, the strategy that keeps both machines busy whenever possible is 2-competitive. The main loss in our derandomization comes from the need to reset.


\subsection{Derandomizing the Makespan Algorithm of Albers \cite{Albers02}}
\label{appendix:albersmakespan}

\begin{theorem}
         Derandomizing the online identical-machines makespan algorithm of Albers \cite{Albers02} yields a deterministic $1.916(4-2\sqrt2)+1\approx 3.245$-competitive algorithm for the online identical-machines makespan problem in the ROM.
\end{theorem}
\begin{proof}
    We begin by greedily spreading identical copies of the first job uniformly across the $m$ machines. Since the copies have equal processing time, the makespan of the most loaded machine when a distinct job arrives is a lower bound on the optimal offline makespan. 

    Upon the arrival of the distinct job, we extract a bit with a worst-case bias of $2-\sqrt2$. We use this bit to select one of the deterministic algorithms from Albers \cite{Albers02}. The choice of which strategy is selected with higher probability is arbitrary. Our derandomization then uses the selected algorithm to process the remainder of the jobs arriving online. When we begin this phase, we make the additional assumption that the machine loads of Albers' algorithm are initialized to 0. This allows the algorithm to ignore the work done in the previous phase, which we later account for with an additive \proc{Opt} term.

    We continue by applying and reapplying the definition of competitive ratios, using the fact that the algorithm of \cite{Albers02} is $1.916$-competitive. The bias from our bit adds at most a $(2-\sqrt2-\frac{1}{2})$ factor to the expected makespan.  Let $\proc{Opt}$ denote the optimal offline makespan for the realized ROM instance. Let $\mathbb{E}[\proc{Alg}]$ denote the expected makespan of our derandomization, and let $\proc{A}_1$ and $\proc{A}_2$ denote the makespans of the deterministic algorithms of Albers \cite{Albers02}. By definition of competitive ratios,
    \begin{align*}
        \mathbb{E}[\proc{Alg}] &\leq \proc{Opt} + (2-\sqrt2)(\proc{A}_1+\proc{A}_2) = \proc{Opt} + \frac{1}{2}(\proc{A}_1+\proc{A}_2)+(2-\sqrt2-\frac{1}{2})(\proc{A}_1+\proc{A}_2) \\
        &\leq \proc{Opt} + 1.916 \, \proc{Opt} + (2-\sqrt2-\frac{1}{2})(2\cdot 1.916 \, \proc{Opt}) \\
        &\leq 3.245 \, \proc{Opt}
    \end{align*}
        
    where the additive \proc{Opt} is due to the initial greedy phase. 
\end{proof}

\end{document}